\documentclass[12pt]{amsart}
\usepackage[latin1]{inputenc}
\usepackage{latexsym}
\usepackage{amsmath}
\usepackage{amssymb}
\usepackage{amsfonts}
\usepackage{graphicx}
\usepackage{epsfig}
\usepackage{color}
\usepackage{bbold}                  
\usepackage{textcomp}               
\usepackage{amstext}
\usepackage{enumerate}
\usepackage{units}
\usepackage{multicol}
 \newtheorem{thm}{Theorem}[section]
 \newtheorem{cor}[thm]{Corollary}
 \newtheorem{lem}[thm]{Lemma}
 \newtheorem{prop}[thm]{Proposition}
 \theoremstyle{definition}
 
 \theoremstyle{remark}

 \numberwithin{equation}{section}
\renewcommand{\le}{\leqslant} 
\renewcommand{\ge}{\geqslant}
\newcommand{\ple}{\prec}

\renewcommand{\imath}{i}

\newcommand{\dom}{\mathrm{dom}}

\newcommand{\one}{\mathbf{1}}
\renewcommand{\L}{\mathcal{L}}
\renewcommand{\d}{\Delta}
\newcommand{\C}{\mathcal{C}}

\newcommand{\R}{\mathcal{R}}

\newcommand{\s}{S}
\newcommand{\bbbone}{\one}
\renewcommand{\r}{R}
\newcommand{\F}{{\mathfrak F}}

\newcommand{\RR}{{\mathbb R}}

\newcommand{\ch}{{\mathcal H}}
\renewcommand{\S}{{\rm S}}
\renewcommand{\R}{{\rm R}}
\renewcommand{\i}{{\rm i}}
\newcommand{\e}{{\rm e}}
\newcommand{\G}{{\mathcal G}}
\newcommand{\scalprod}[2]{\langle{#1},{#2}\rangle}

\DeclareGraphicsExtensions{.pgn}

\begin{document}

\title{On the irreversible dynamics emerging \\
from quantum resonances}

\address[M.~K\"onenberg]{Department of Mathematics and Statistics\\ Memorial University, St. John's, NL, Canada}
\author{M. K\"onenberg}
\curraddr{Fachbereich Mathematik\\ Universit\"at Stuttgart,  Stuttgart, Germany} 
\email{martin.koenenberg@mathematik.uni-stuttgart.de}
\address[M.~Merkli]{Department of Mathematics and Statistics\\ Memorial University, St. John's, NL, Canada}
\author{M. Merkli}
\email{merkli@mun.ca}
%


\date{\today}

\begin{abstract}
We consider the dynamics of quantum systems which possess stationary states as well as slowly decaying, metastable states arising from the perturbation of bound states. We give a decomposition of the propagator into a sum of a stationary part, one exponentially decaying in time and a polynomially decaying remainder. The exponential decay rates and the directions of decay in Hilbert space are determined, respectively, by complex resonance energies and by projections onto resonance states. Our approach is based on an elementary application of the Feshbach map. It is applicable to open quantum systems and to situations where spectral deformation theory fails. We derive a detailed description of the dynamics of the spin-boson model at  arbitrary coupling strength. 
\end{abstract}

\maketitle

\section{Introduction and main result}

\subsection{General setup}

Let $L_0$ be a self-adjoint operator on a Hilbert space $\mathcal H$ and consider 
\begin{equation}
\label{2}
L=L_0+\d I,
\end{equation}
where $\d\in\mathbb R$ is a perturbation parameter and $I$ is a self-adjoint operator. It is assumed that  $L$ is self-adjoint. We suppose that the spectrum of $L_0$ is absolutely continuous (possibly, but not necessarily semi-bounded) and that $L_0$ has finitely many eigenvalues $e$ with finite multiplicities $m_e$. All the eigenvalues of $L_0$ are embedded in the continuous spectrum. The general problem we consider is how the stability or partial stability or instability of these embedded eigenvalues under the perturbation affect the dynamics generated by $L$. One may readily incorporate into our results and proofs the case where $L_0$ has also isolated eigenvalues by using ordinary analytic perturbation theory on the corresponding subspaces.

In the setting of usual analytic perturbation theory \cite{Kato} an isolated eigenvalue $e$ of $L_0$ with  multiplicity $m_e$ splits, under perturbation, into a group of eigenvalues $E_{e,1},\ldots,E_{e,\ell_e}$ of $L$ ($1\le \ell_e\le m_e$), in the sense that $E_{e,j}=E_{e,j}(\d)\rightarrow e$ as $\d\rightarrow 0$, for $j=1,\ldots,\ell_e$. For fixed $e$, the sum of the multiplicities of the eigenvalues $E_{e,j}$ equals $m_e$. On the other hand, it is well known that {\em embedded} eigenvalues can be unstable, or partially stable, under perturbation. Instability means that $L$ does not have any eigenvalues in a neighbourhood of $e$ for small $\d$. Partial stability means that the embedded eigenvalue $e$ of $L_0$ splits into a group of eigenvalues of $L$ whose sum of multiplicities are strictly smaller than that of $e$. 

In this paper, we consider the situation where $L_0$ has unstable and partially stable eigenvalues, and where  the partially stable ones undergo a reduction to dimension one under perturbation. Namely, close to any eigenvalue $e$ of $L_0$, the operator $L$ either does not have any eigenvalue ($e$ unstable) or $L$ has exactly one simple eigenvalue $E_e$ close to $e$, meaning that $\lim_{\d\rightarrow 0}E_e=e$. It is supposed that all eigenvalues of $L$ are of this form. One may develop the arguments of the present paper in the more general setting where close to every $e$, $L$ has several eigenvalues $E_1,\ldots,E_\ell$ and each of them may be degenerate.  We do not do this here to keep the exposition simpler. 

The dependence of $E_e$ on $\d$ is not governed by usual analytic perturbation theory, since the unperturbed $e$ is an embedded eigenvalue of $L_0$. However, suitably modified expressions from analytic perturbation theory of isolated eigenvalues will still play a role in the present setting.  Let $P_e$ be the spectral projection of $L_0$ associated to the eigenvalue $e$. If $e$ was an isolated eigenvalue of $L_0$, then the first and the second order corrections (in $\d$) of eigenvalues would be given, according to analytic perturbation theory \cite{Kato}, by the eigenvalues of $P_e IP_e$ and of $P_e IP_e^\perp(L_0-e)^{-1}IP_e$, respectively.  

We assume that 
\begin{itemize}
\item[{\bf (A1)}] For all eigenvalues $e$ of $L_0$,
\begin{equation}
\label{n1}
P_e I P_e =0.
\end{equation}
\end{itemize}
{}For embedded $e$, the resolvent $P^\perp_e(L_0-e)^{-1}$ does not exist as a bounded operator, so $P_e IP_e^\perp(L_0-e)^{-1}IP_e$ is not defined, typically. Nevertheless, we can replace $e$ by $e-\i\epsilon$ and consider $\epsilon$ small. This suggests that the second order eigenvalue corrections to $e$ are linked to the {\em level shift operator}
\begin{equation}
\label{n2}
\Lambda_e = -P_e IP_e^\perp (L_0-e+\i0_+)^{-1} IP_e,
\end{equation}
where $\i0_+$ indicates the limit 
of the resolvent $(L_0-e+\i \epsilon)^{-1}$,
as $\epsilon\rightarrow 0_+$. The existence of the limit is guaranteed by assumption (A2)
below (take $\d=0$ in the resolvent in \eqref{nlap}). The operator $\Lambda_e$ is represented by an $m_e\times m_e$ matrix.

Let $Q$ be an orthogonal projection and denote
\begin{equation}
\label{n4}
R_z=(L-z)^{-1}\quad \mbox{and}\quad R_z^{Q}= (Q^\perp L Q^\perp-z)^{-1}\upharpoonright_{{\rm Ran} Q^\perp} . 
\end{equation} 
In the following, we denote by $C(\phi,\psi)$ a constant which is independent of $z$ and $\Delta$ but which may depend on $\phi, \psi\in \mathcal H$.
\begin{itemize}
\item[{\bf (A2)}] (Limiting Absorption Principle.)  
There is a dense set ${\mathcal D}\subset{\mathcal H}$ with ${\rm Ran}\, IP_e\subset\mathcal D$ ($\forall e$) and there is an $\alpha>0$ such that the following hold.\\
{\bf (1)} Let $S_e=\{ z\in{\mathbb C}_-\ :\ |{\rm Re}z-e|\le\alpha\}$.
Here, ${\mathbb C}_-$ denotes the (open) lower complex half  plane. For all $e$ and all $\phi,\psi\in{\mathcal D}$, we have 
\begin{equation}
\label{nlap}
\sup_{z\in S_e}| \tfrac{d^k}{dz^k} \scalprod{\phi}{R_z^{P_e}\psi}| \le  C(\phi,\psi)<\infty,\ \quad k=0,\ldots 3.
\end{equation}
{\bf (2)} Let $S_\infty=\{z\in{\mathbb C}_-\ :\ |{\rm Re}z-e|>\alpha \mbox{ for all $e$}\}$. For all $\phi,\psi\in{\mathcal D}$, we have 
\begin{equation}
\label{nlap1}
\sup_{z\in S_\infty}| \tfrac{d^k}{dz^k} \scalprod{\phi}{ R_z\psi}| \le  C(\phi,\psi)<\infty,\ \quad k=0,1. 
\end{equation}
\end{itemize}
Assumption  (A2)(1)
implies that $P_e I R_z^{P_e} IP_e$ extends as a twice continuously differentiable function 
to real $z$ with $|z-e|\le\alpha$. Moreover, the estimates \eqref{nlap} and \eqref{nlap1} with $k=0$ imply that $P^\perp_e LP^\perp_e$ and $L$ have purely absolutely continuous spectrum in the interval $(e-\alpha,e+\alpha)$ and in the region $\{x\in{\mathbb R}\ :\ |x-e|>\alpha \mbox{ for all $e$}\}$, respectively (see e.g. \cite{KoMeSo} or \cite{CFKS}, Proposition 4.1).
\begin{itemize}
\item[({\bf A3})] The operators $\tfrac{d}{d \Delta} P_e I R_z^{P_e} IP_e$  is bounded,
uniformly for $z\in{\mathbb C}_-$ with $|{\rm Re}z-e|\le \alpha$ and for all $|\Delta|\le \Delta_0$.
\end{itemize}
The Feshbach map associated to an orthogonal projection $Q$, applied to $L-z$, is defined by
\begin{equation}
\label{n3}
\F(L-z;Q)= Q( L-z- L R_z^Q L)Q.
\end{equation}
It follows from  assumptions (A1) and (A2)
that 
$$
\lim_{\epsilon\rightarrow 0_+} \F(L-E_e+\i\epsilon;P_e) \equiv \F(L-E_e;P_e) = P_e(e-E_e-\d^2 IR_{E_e-\i0_+}^{P_e}I)P_e.
$$
Let $\psi_{E_e}=\psi_{E_e}(\d)$ be s.t. $L\psi_{E_e}=E_e\psi_{E_e}$. Note that $P_e\psi_{E_e}\neq 0$ for otherwise $P^\perp_eLP^\perp_e\psi_{E_e}=E_e P_e^\perp\psi_{E_e}$, which cannot hold for small $\d$, since $P^\perp_eLP^\perp_e$ has purely 
absolutely continuous spectrum in a neighbourhood of $e$ due to Condition (A2)(1). We normalize $\psi_{E_e}$ as $\|P_e\psi_{E_e}\|=1$. 

By the isospectrality property of the Feshbach map (see Appendix \ref{FeshbachAppendix}, Proposition \ref{Feshprop}), we have 
\begin{equation}
\label{n5}
0=\F(L-E_e;P_e) P_e\psi_{E_e}=\big(\tfrac{e-E_e}{\d^2}-P_eIR_{E_e-\i 0_+}^{P_e}IP_e)P_e\psi_{E_e}.
\end{equation}
This, together with Conditions (A2)(1) and (A3), implies that $\xi_\d:= \tfrac{e-E_e}{\d^2}$ is bounded in $\d$ for small $\d$ and that $P_eIR_{E_e-\i 0_+}^{P_e}IP_e=\Lambda_e+O(|\d|+ |e-E_e|)=\Lambda_e+O(|\d|)$. On a suitable sequence $\d_n\rightarrow 0$, we have $\xi_{\d_n}\rightarrow \xi_0$ and $P_e\psi_{E_e}(\d_n)\rightarrow P_e\psi_0$  for some $\xi_0\in\mathbb R$ and some unit vector $P_e\psi_0$ (Bolzano-Weierstrass). Consequently, taking $\d\rightarrow 0$ in \eqref{n5} along this sequence gives
\begin{equation}
\label{nnn1}
\Lambda_eP_e\psi_0 = \xi_0P_e\psi_0,
\end{equation}
showing that $\xi_0$ is a real eigenvalue of $\Lambda_e$.  For ease of presentation, we assume the following. 

\medskip

\begin{itemize}
\item[{\bf (A4)}] (Fermi Golden Rule Condition.) The eigenvalues of all the level shift operators $\Lambda_e$ are simple. Moreover,

{\bf (1)} If $e$ is an unstable eigenvalue, then all the eigenvalues $\lambda_{e,0},\ldots,\lambda_{e,m_e-1}$ of $\Lambda_e$ have strictly positive imaginary part.

{\bf (2)} If $e$ is a partially stable eigenvalue, then $\Lambda_e$ has a single real eigenvalue $\lambda_{e,0}$. All  other eigenvalues  $\lambda_{e,1},\ldots,\lambda_{e,m_e-1}$ have strictly positive imaginary part.
\end{itemize}

\medskip

Under condition (A4)(2), the set $\{\xi_\d=\tfrac{e-E_e}{\d^2}\}$ for $\d$ small has a unique limit point $\xi_0$ and we have $\lambda_{e,0}=\xi_0$.
Having only simple eigenvalues, $\Lambda_e$ is diagonalizable and has the spectral representation
\begin{equation}
\label{n7}
\Lambda_e =\sum_{j=0}^{m_e-1} \lambda_{e,j}P_{e,j},
\end{equation}
where $P_{e,j}$ are the (rank one) spectral projections. We introduce the notation
\begin{equation}
\label{defple}
a\ple b,
\end{equation}
where $a$ is a complex number, a vector or a bounded operator and $b>0$, to mean that $|a|\le {\rm const.} b$, where $|\cdot|$ is the appropriate norm and $\rm const.$ is a constant which does not depend on the coupling parameter $\d$, nor on time $t$.

\begin{thm}[Resonance expansion of propagator]
\label{nthm1}
There is a constant $c>0$ s.t. for $0<|\Delta|<c$ the following holds. Denote the projection onto the eigenvalue $E_e$ of $L$ by $\Pi_{E_e}$. Let $t>0$, $\phi,\psi\in\mathcal D$ s.t. $L\phi$, $L\psi\in{\mathcal D}$. Then
\begin{eqnarray}
\label{nmainreseqn}
\scalprod{\phi}{\e^{\i t L}\psi} &=& \sum_{e\ {\rm partially\,  stable}}\Big\{  \e^{\i t E_e} \scalprod{\phi}{\Pi_{E_e}\psi} + \sum_{j=1}^{m_e-1} \e^{\i t (e+ \d^2 a_{e,j})}\scalprod{\phi}{\Pi'_{e,j}\psi} \Big\}\\
&& + \sum_{e\ {\rm unstable}} \sum_{j=0}^{m_e-1} \e^{\i t (e+ \d^2 a_{e,j})}\scalprod{\phi}{\Pi'_{e,j}\psi}  +R(t),
\nonumber
\end{eqnarray}
where
\begin{equation}
\label{nmthm1}
R(t) \ple \frac{1}{t}.
\end{equation}
The exponents $a_{e,j}$ and the operators $\Pi'_{e,j}$ are close to the spectral data of the level shift operator $\Lambda_e$, \eqref{n7}. Namely, 
\begin{equation}
\label{nmthm2}
a_{e,j} = \lambda_{e,j}+O(\d), \qquad \Pi'_{e,j} =P_{e,j}+O(\d).
\end{equation}
\end{thm}

\medskip
\noindent
{\bf Remarks and discussion.\ } 
\smallskip

1. For an expansion of $\scalprod{\phi}{\e^{-\i tL}\psi}$ for $t>0$, simply take the adjoint of \eqref{nmainreseqn}. 

2. The exponents $a_{e,j}$ are the eigenvalues of an explicit matrix. They can be calculated to all orders in $\d$ (see Lemma \ref{lemma1}). The operators $\Pi'_{e,j}$ have also expressions calculable to all orders in $\d$ (see \eqref{qtilde}). Those ``complex energies'' $a_{e,j}$ are called resonances. They coincide with the eigenvalues of the spectrally deformed generator of dynamics in situations where the latter exists.

3. The remainder term is small relative to the contributions of the exponentially decaying terms in \eqref{nmainreseqn} for times $t$ satisfying $\e^{-\gamma \d^2 t}>\!\!>1/t$, where
\begin{equation}
\label{decayrate}
\gamma=\min_{e,j}\big\{{\rm Im} a_{e,j}  \big\}.
\end{equation}
The inequality $\e^{-\gamma \d^2 t}\ge C/t$, for some (large) $C$ is equivalent to $\frac{\ln(t)-\ln(C)}{t}\ge\gamma\d^2$. For small $\d$,  it is valid for {\em intermediate times}, $t_0 < t < t_1$, with $t_0=C+O(\gamma\d^2)$ and  $t_1\sim 1/(\gamma\d^2)$. During this time-interval, the decay of \eqref{nmainreseqn} behaves as exponential, to leading order.

4. The Fermi Golden Rule Condition (A4) guarantees that instability (and partial stability) of eigenvalues is visible at lowest order, $O(\Delta^2)$, in the perturbation. It may happen that resonances acquire non-vanishing imaginary parts only at higher orders in $\Delta$. Our method can be adapted to describe this situation. The dynamical consequence is a slower decay of the corresponding directions in Hilbert space, see also \cite{MeBeSo}.

5. Increasing the regularity assumptions on the vectors $\phi$, $\psi$ allows to show a faster polynomial decay of the remainder than \eqref{nmthm1} (this amounts to taking higher $z$ derivatives of the resolvent, c.f. \eqref{14}).

\subsection{History, relation to other work.}

The analysis of resonance phenomena has a long history and plays an important role in quantum physics \cite{BW,KP,S,WW}. Its modern description, involving dilation analytic Hamiltonians (\cite{AC,BC}), was given in \cite{Simon} and further developed in \cite{Simon2,Sig1,Hun2}. We refer to \cite{HiSig} for a textbook presentation and many more references. A time-dependent theory of quantum resonances was established in \cite{MS}, inspired by \cite{SW} and further developed in \cite{KRW}. In these works, as well as in \cite{CGH}, a variant of the Mourre theory in combination with the Feshbach projection method is used to link dynamical properties of quantum systems to spectral objects. The approach of the present work is, in spirit, similar to \cite{CGH}; see the end of this paragraph for a comparison. Nevertheless, all the above-mentioned works require regularity conditions within Mourre theory that do not allow the treatment  of open quantum systems at positive temperature.  In the context of open quantum systems at positive temperature, the link between quantum resonances and approach of an equilibrium state has been pioneered, using complex deformation theory, in \cite{JP,BFSrte}. The work \cite{BFSrte} is based on a sophisticated renormalization group method initiated in \cite{BFS1,BFS2,BFS3}. Recently, a method based on graph expansions of the propagator rather than purely spectral considerations was given in \cite{DeKu}. The spectral approach has been further developed to yield a detailed description of open systems dynamics in terms of resonances in \cite{MSB1,MSB2,MSB3}, with applications to quantum information theory  \cite{M,MBBG} and quantum chemistry \cite{MBS}. The spectral analysis and its consequences for ``return to equilibrium'' based on Mourre theory and positive commutators was carried out in  \cite{DJ,Mthesis,FM,DJ2}. However, these papers are limited to the study of the spectrum of the Liouville operator with the goal (typically) of showing that it has a single, simple eigenvalue at zero (and absolutely continuous spectrum otherwise). This information alone does not provide any detail about the dynamics other than ergodicity. However, one is interested in information such as directions of decay and decay rates which describe, for example the speed of thermalization, decoherence and the dynamics of entanglement.  In the method of complex deformation, complex resonance energies are linked ``automatically'' to the decay rates of reduced density matrix elements \cite{MSB2}. The same expressions describing those decay rates appear as well in Mourre theory as a consequence of the Fermi Golden Rule (see also \cite{FaMoSk}), however, linking them to the dynamics, and in particular to time decay, is more delicate and has not been done previously for open systems.  We show in the present paper how to extract the detailed dynamical information from the Mourre theory in a technical setup that includes positive temperature open quantum systems. It is important to be able to handle these questions using a softer approach than the spectral deformation one. Indeed, the applicability of the latter demands much more regularity from the models and, in some physically relevant situations, the spectral deformation technique is not  applicable at all. This happens for the spin-boson model at arbitrary coupling, whose ergodicity has been shown recently in \cite{KoMeSo} using Mourre theory. As an application of our method, we give a detailed expansion of the propagator of this model in the present paper. 

The philosophies of \cite{CGH} and the present paper are similar, in that the common main idea is to write the propagator as a contour integral over the resolvent and subsequently use the Feshbach map to analyze the latter. However, right from the start, the technical assumptions are very different. A core assumption of \cite{CGH} is that multicommutators of the Hamiltonian $H$ with the (Mourre theory) conjugate operator $A$ are relatively $H$-bounded. While this is typically true for, say, for Schr\"odinger operators, it is not so for open quantum systems. The problem comes from the fact that the number operator is not bounded relatively to the free field Hamiltonian (so already relation (2) of \cite{CGH} is not valid). The situation even is worse for the spin-boson model at arbitrary coupling, where each successive commutation of the Liouvillian with the conjugate operator produces a more singular operator, as explained in \cite{KoMeSo} (this is the reason why the spectral deformation theory fails).

The dissipative character of the system is guaranteed in \cite{CGH} by assuming  a Mourre estimate, localized spectrally on a subspace around the embedded eigenvalue in question. Accordingly, the main result of \cite{CGH} describes the dynamics of an initial state (wave function) which is spectrally localized close to the embedded eigenvalue. This makes good physical sense in the context of, say, Schr\"odinger operators, where long lived initial states are expected to lie close to unperturbed bound states. However, in open systems problems, one considers initial states which are spatial perturbations of equilibrium states and which are not at all spectrally localized relative to the Liouville operator. Our dissipation assumption (A2) is thus a Limiting Absorption Principle which is {\em not} spectrally localized and which produces results for initial vectors which are not spectrally localized.  Condition (A2) can be heuristically understood in the context of open systems as saying that the reservoir stays essentially in its equilibirum state during the dynamical process (the Born approximation).

We mention that the main result of \cite{CGH} is stated and proven for a simple unperturbed eigenvalue. In open systems however, the origin is {\em always} a degenerate eigenvalue of the unperturbed Liouville operator and so we have put in place a formalism that works for the degenerate case as well.

\subsection{Outline of the proof of Theorem \ref{nthm1}}

For  $\psi\in\dom(L)$, we have (\cite{EngelNagel}, Corollary II 3.6)
\begin{equation}
\label{4}
\e^{\i tL} \psi= -\frac{1}{2\pi\i}\int_{{\mathbb R}-\i w}\e^{\i tz}R_z\psi \,dz,
\end{equation}
where $w>0$ is arbitrary and $R_z=(L-z)^{-1}$, see \eqref{n4}.   We write $\e^{\i tz}=\frac{1}{\i t}\frac{d}{dz}\e^{\i tz}$ and integrate by parts in \eqref{4} to obtain
\begin{equation}
\label{14}
e^{\i tL} \psi= \frac{1}{\i t}\frac{1}{2\pi\i}\int_{{\mathbb R}-\i w}\e^{\i tz}\textstyle\frac{d}{dz}R_z\psi \,dz.
\end{equation}
The boundary terms in the integration by parts vanish since $R_z\psi\rightarrow 0$ as $|z|\rightarrow\infty$. 

We analyze separately the contributions to the integral in \eqref{14} coming from $z$ in different regions on the line of integration. Define the gap of all the clusters of resonances by
\begin{equation}
\label{delta}
\delta = \min_{e,i,j}\big\{ |\lambda_{e,i}-\lambda_{e,j}|\ :\ i\neq j \big\}>0
\end{equation}
and denote the eigenvalue gap of $L_0$ by
\begin{equation}
\label{gapdef}
g = \min_{e,e'}\{ |e-e'|\ :\  e\neq e'\} > 0.
\end{equation}
Let
\begin{equation}
\label{13}
\alpha=\tfrac12\min\big\{c \delta,\, g, \, {\rm Im}\lambda_{e,j}\ :\ (e,j) \ {\rm s.t.\ } {\rm Im}\lambda_{e,j}>0  \big\} > 0.
\end{equation}
Here, $c$ is a constant not depending on $\d,w$ (its origin is explained in Lemma \ref{lemma1}).
For any eigenvalue $e$ of $L_0$, set
\begin{equation}
\G_e=\{x-\i w\ :\ |x-e|\le \alpha\}
\label{6}
\end{equation}
and set
\begin{equation}
 \G_\infty= \big\{x-\i w\ :\ x\in{\mathbb R}\big\} \backslash \cup_{e} \G_e.
\label{8}
\end{equation}
It follows from \eqref{14} that 
\begin{equation}
\label{9}
\scalprod{\phi}{\e^{\i tL}\psi}= \sum_{e} J_e(t) +J_\infty(t), 
\end{equation}
where
\begin{equation}
\label{10}
J_\#(t) =\frac{1}{\i t}\frac{1}{2\pi\i}\int_{\G_\#}\e^{\i t z} \scalprod{\phi}{\textstyle\frac{d}{dz}R_z\psi} dz.
\end{equation}
We now apply a suitable Feshbach map to the resolvent $R_z$ in \eqref{10}, with a projection depending on the region of integration. Let $P$ be an orthogonal projection and recall the notation \eqref{n4}. The resolvent has the representation
\begin{equation}
\label{0} 
R_z=\F(z)^{-1} +{\mathcal B}(z)+ R_z^Q, 
\end{equation}
where $\F(z)\equiv \F(L-z,Q)$, see \eqref{n3},
%
 and
\begin{equation}
{\mathcal B}(z) = - \F(z)^{-1} Q L  R^Q_z  - R_z^Q L Q\F(z)^{-1} +R_z^Q L Q \F(z)^{-1} Q L R_z^Q.
\label{0.1}
\end{equation}
We explain these relations and some properties of the Feshbach map in Appendix \ref{FeshbachAppendix}. For $z\in\G_e$, we choose the projection $Q$ in the Feshbach map to be $P_e$. For $z\in\G_\infty$, the argument is simpler, see Section \ref{inftysect}. 

Let us assume that the unperturbed, partially stable eigenvalue of $L_0$ is at the origin, $e=0$. (Otherwise see section \ref{proofthmnthm1sect}.) Then $L$ has a simple eigenvalue $E\equiv E_0$
 with $E\rightarrow 0$ as $\d\rightarrow 0$.  To analyze $J_0(t)$, we write, according to \eqref{0}, 
\begin{equation}
\label{p8}
J_0(t) = \frac{1}{\i t}\frac{1}{2\pi\i} \int_{\G_0} \e^{\i tz} \left\{ \scalprod{\phi}{\tfrac{d}{dz}\F(z)^{-1}\psi}
+\scalprod{\phi}{\tfrac{d}{dz}{\mathcal B}(z)\psi}
+\scalprod{\phi}{\tfrac{d}{dz}R_z^{P_0}\psi} \right\} dz.
\end{equation}
The Feshbach term is $\F(z) = -z+\d^2A_z$, where $A_z=-P_0IR_z^{P_0}IP_0$. For $z=0$ and $\d=0$, $A_z$ is just the level shift operator $\Lambda_0$, \eqref{n2}. We show in Lemma \ref{lemma1} that $A_z$ is diagonalizable,
$A_z = \sum_{j=0}^{m_e-1} a_{0,j}(z)Q_j(z)$, and that the eigenvalues $a_{0,j}(z)$ of $A_z$ satisfy $a_{0,0}(E)=E/\d^2$
 for all $\d\neq 0$ (this follows from the isospectrality property of the Feshbach map and the fact that $E$ is an eigenvalue of $L$) and  $a_{0,j}(z) = \lambda_{0,j}+O(\d^2+|z|)$, $j=1,\ldots,m_e-1$ (since $A_z$ is close to $\Lambda_0$). Then we can write 
\begin{equation}
\label{p3}
\tfrac{d}{dz}\F(z)^{-1} = \tfrac{d}{dz} \sum_{j=0}^{m_e-1} \frac{Q_j(z)}{-z+\d^2 a_{0,j}(z)}.
\end{equation}
We are interested in the singularities of this function as $z$ is close to the real axis. They come from the denominator. To understand the nature of the singularities, and since $z\mapsto Q_j(z)$  is regular, consider $Q_j(z)\approx Q_j(0)$ 
for a moment. Then
\begin{equation}
\label{p4}
\tfrac{d}{dz}\F(z)^{-1}\approx \sum_{j=0}^{m_e-1} \frac{1-\d^2a'_{0,j}(z)}{(z-\d^2a_{0,j}(z))^2} Q_j(0).
\end{equation} 
For $j=0$ we have $-E+\d^2a_{0,0}(E)=0$ (see above) and the corresponding summand is 
\begin{eqnarray*}
\lefteqn{
\frac{1}{(z-E)^2} \frac{1-\d^2a'_{0,0}(z)}{\{1-\d^2[a_{0,0}(z)-a_{0,0}(E)]/(z-E)\}^2} Q_0(0)}\\
&& \qquad\qquad\qquad\approx \frac{1}{(z-E)^2} \frac{1-\d^2a'_{0,0}(E)}{(1-\d^2a'_{0,0}(E))^2} Q_0(0) = \frac{1}{(z-E)^2} \frac{Q_0(0)}{1-\d^2a'_{0,0}(E)}.
\end{eqnarray*}
By using that the projection associated to the eigenvalue $E$ of $L$ is given by
\begin{equation}
\label{p12}
\Pi_E=\lim_{\epsilon\rightarrow 0_+}(\i\epsilon)(L-E+\i\epsilon)^{-1}
\end{equation}
and decomposing the resolvent in this limit according to \eqref{0} with projection $P_0$, we identify (see \eqref{034}) 
$$
\frac{Q_0(E)}{1-\d^2a'_{0,0}(E)} = P_0\Pi_EP_0.
$$
For $j>0$ we have $a_{0,j}(0)=\lambda_{0,j}+O(\d^2)$ which is in the open upper complex half plane and the corresponding summand in \eqref{p4} is 
$$
 \frac{1-\d^2a'_{0,j}(z)}{(z-\d^2a_{0,j}(z))^2} Q_j(0) \approx \frac{Q_j(0)}{(1-\d^2a_{0,j}(0))^2}.
$$
In Section \ref{proofthmnthm1sect} we make these arguments rigorous. Namely, we show that
\begin{equation}
\label{p2}
\tfrac{d}{dz}\F(z)^{-1} =\frac{1}{(z-E)^2}P_0\Pi_EP_0 +\sum_{j=1}^{m_e-1} \frac{Q_j(0)}{(z-\d^2 a_{0,j}(0))^2} + \widetilde T(z),
\end{equation}
where $\int_{\G_0}\e^{\i tz} \widetilde T(z)dz\ple 1$. Now we have to multiply  \eqref{p2} by $\e^{\i t z}$ and integrate over $z\in\G_0 = [-\alpha,\alpha]-\i w$. Having in mind a standard argument from complex analysis, we complete the path $\G_0$ into a closed contour (a rectangle with a `roof' parallel to $\G_0$ but shifted far into the upper complex half plane). We then use the Cauchy formula for contour integrals to get
\begin{equation}
\label{p5}
\int_{\G_0}\frac{\e^{\i tz}}{(z-E)^2}dz = 2\pi\i\cdot \i t \e^{\i t E}+O(1/t). 
\end{equation}
The $O(1/t)$ term is the contribution from the parallel vertical sides of the rectangular closed integration path (see \eqref{017}). In a similar way, we treat the sum in \eqref{p2}. Here the poles are at $z=\d^2 a_{0,j}(0)$ and so
\begin{equation}
\label{p6}
\int_{\G_0}\frac{\e^{\i tz}}{(z-\d^2a_{0,j}(0))^2} dz = 2\pi\i \cdot\i t\, \e^{\i t\d^2a_{0,j}(0)} +O(1/t).
\end{equation}
Combining \eqref{p2} with \eqref{p5} and \eqref{p6} yields 
\begin{equation}
\label{p7}
\frac{1}{\i t}\frac{1}{2\pi\i} \int_{\G_0}\e^{\i tz} \scalprod{\phi}{\tfrac{d}{dz}\F(z)^{-1}\psi} dz = \e^{\i t E}\scalprod{\phi}{P_0\Pi_EP_0\psi} +\sum_{j=1}^{m_e-1}  \e^{\i t\d^2 a_{0,j}(0)} \scalprod{\phi}{Q_j(0)\psi} +O(1/t). 
\end{equation}

Next we deal with the second integrand in \eqref{p8}. Using again the spectral representation of $\F(z)^{-1}$, we have from \eqref{0.1}
\begin{equation}
\label{p10}
\tfrac{d}{dz}{\mathcal B}(z) = \tfrac{d}{dz} \sum_{j=0}^{m_e-1} \frac{q_j(z)}{-z+\d^2a_{0,j}(z)},
\end{equation}
where
\begin{equation}
\label{p9}
q_j(z) =-\d \left[ Q_j(z)P_0IR_z^{P_0} +R_z^{P_0}IP_0Q_j(z) -\d R_z^{P_0} IP_0Q_j(z)P_0IR_z^{P_0}\right].
\end{equation}
The expression \eqref{p10} has the same structure as \eqref{p3}. We readily obtain, in analogy with \eqref{p7},
\begin{equation*}
\label{p11}
\frac{1}{\i t}\frac{1}{2\pi\i} \int_{\G_0}\e^{\i tz} \scalprod{\phi}{\tfrac{d}{dz}{\mathcal B}(z)\psi} dz =\e^{\i tE} \scalprod{\phi}{\tfrac{q_0(E)}{1-\d^2a'_{0,0}(E)}\psi} +\sum_{j=1}^{m_e-1}  \e^{\i t\d^2 a_{0,j}(0)} \scalprod{\phi}{q_j(0)\psi} +O(|\d|/t). 
\end{equation*}
Proceeding as above, after \eqref{p12}, we identify $P_0\Pi_EP_0 +\tfrac{q_0(E)}{1-\d^2a'_{0,0}(E)}=\Pi_E$ (see also \eqref{039}). Finally, since by Assumption (A2)(1),
$$
\int_{\G_0}\e^{\i tz}\scalprod{\phi}{\tfrac{d}{dz}R_z^{P_0}\psi} dz\ple 1,
$$
we obtain 
\begin{equation}
\label{p13}
J_0(t) = \scalprod{\phi}{\Pi_E\psi} + \sum_{j=1}^{m_e-1}\e^{\i t\d^2 a_{0,j}(0)}\scalprod{\phi}{\Pi'_{0,j}\psi}  +O(1/t),
\end{equation}
where $\Pi'_{0,j} = Q_0(0)+O(|\d|)=P_{0,j}+O(|\d|)$. This explains the contribution of a term on the right side of \eqref{nmainreseqn} coming from a partially stable eigenvalue $e$ ($=0$). The analysis for unstable $e$ follows using the same arguments. Finally, to deal with $J_\infty(t)$, we write
\begin{equation}
\label{i0}
R_z\psi = (z+\i)^{-2} R_z(L+\i)^2\psi -(z+\i)^{-1}\psi - (z+\i)^{-2}(L+\i)\psi,
\end{equation}
which is valid for $\psi\in\dom(L^2)$. The negative powers of $z$ help the convergence of the $z$-integral over $\G_\infty$. The bound $J_\infty(t)\ple1/t$ is then easily reached using (A2)(2).

\section{Application to open quantum systems}

\subsection{Setup} The Hilbert space is the product of a system and a reservoir part,
\begin{equation}
\label{1}
\ch =\ch_\S\otimes\ch_\R.
\end{equation}
The self-adjoint generator of dynamics, called {\em Liouvillean}, is of the form \eqref{2}, where $L_0$, the free (non interacting) Liouvillean, is a sum of a system and a reservoir contribution,
\begin{equation}
\label{3}
L_0=L_\S+L_\R,
\end{equation}
and $I$ is the system-reservoir interaction operator. We consider the system to be finite-dimensional and the reservoir to be an infinitely extended free Bose gas at positive temperature, as we explain now.

Let $\S$ be a quantum system with pure state space ${\mathfrak H}_\S$ of dimension $d_0<\infty$. For instance, for a spin $1/2$, $d_0=2$. Then the Hilbert space ${\mathcal H}_\S$ in \eqref{1} is the GNS space (Liouville space)
\begin{equation}
\label{particleGNS}
{\mathcal H}_\S = {\mathfrak H}_\S\otimes{\mathfrak H}_\S,
\end{equation}
so that $d=\dim {\mathcal H}_\S=d_0^2$. The doubling of the pures state system Hilbert space in \eqref{particleGNS} allows to represent any (pure or mixed) state of $\S$ by a vector. Namely, let $\rho$ be a density matrix on ${\mathfrak H}_\S$. It has the diagonalized form $\rho=\sum_i p_i |\psi_i\rangle\langle\psi_i|$, to which we associate the vector $\Psi_\rho = \sum_i \sqrt{p_i}\psi_i\otimes\overline\psi_i\in {\mathfrak H}_\S\otimes {\mathfrak H}_\S$ 
(complex conjugation in any fixed basis -- we will choose the eigenbasis of the system Hamiltonian). Then ${\rm Tr}(\rho A)=\scalprod{\Psi_\rho}{(A\otimes\one_\S)\Psi_\rho}$ for all $A\in{\mathcal B}({\mathfrak H}_\S)$ and where $\one_\S$ is the identity in ${\mathfrak H}_\S$. This is the GNS representation of the state given by $\rho$ \cite{BR,MSB2}. Let $H_\S=\sum_j E_j |\varphi_j\rangle\langle\varphi_j|$ be the Hamiltonian of $\S$, acting on ${\mathfrak H}_\S$. The equilibrium density matrix is $\rho_\S = \e^{-\beta H_\S}/{\rm Tr\,}\e ^{-\beta H_\S}$, which is represented on ${\mathcal H}_\S$ by the vector
\begin{equation}
\label{systemkms}
\Omega_{\S,\beta} = ({\rm Tr\,}\e^{-\beta H_\S})^{-1} \sum_j \e^{-\beta E_j/2} \varphi_j\otimes\varphi_j.
\end{equation}

The (GNS) Hilbert space of the spatially infinitely extended free bose gas, for states normal w.r.t. the equilibrium (KMS) state, is the Fock space 
\begin{equation}
{\mathcal H}_\R = {\mathcal F}_\beta= \bigoplus_{n\geq 0} L^2_{\rm sym}(({\mathbb R}\times S^2)^{n},(d u\times d\Sigma)^{n}),
\label{2.1}
\end{equation}
taken over the single-particle space $L^2({\mathbb R}\times S^2,d u\times d\Sigma)$, where $d\Sigma$ is the uniform measure on $S^2$ \cite{AW,JP}. ${\mathcal F}_\beta$ carries a representation of the CCR algebra in which the Weyl operators are given by $W(f_\beta) = e^{i\phi(f_\beta)}$, where $\phi(f_\beta)=\frac{1}{\sqrt{2}}(a^*(f_\beta)+a(f_\beta))$. Here, $a^*(f_\beta)$ and $a(f_\beta)$ denote creation and annihilation operators on ${\mathcal F}_\beta$, smoothed out with the function
\begin{equation}
f_\beta(u,\Sigma) = \sqrt{\frac{u}{1-e^{-\beta u}}}\ |u|^{1/2} \left\{
\begin{array}{ll}
f(u,\Sigma), & u\geq 0\\
-\overline{f}(-u,\Sigma), & u<0
\end{array}
\right.
\label{2.3}
\end{equation}
belonging to $L^2({\mathbb R}\times S^2,d u\times d\Sigma)$. It is easy to see that the CCR are satisfied, namely,
\begin{equation}
W(f_\beta)W(g_\beta) = e^{-\frac{i}{2}{\rm Im}\scalprod{f}{g}} W(f_\beta+g_\beta).
\label{ccr}
\end{equation}
The vacuum vector $\Omega\in{\mathcal F}_\beta$ represents the infinite-volume equilibrium state of the free Bose field, determined by the formula 
\begin{equation}
\label{thav}
\scalprod{\Omega}{W(f_\beta)\Omega} = \exp\left\{ \textstyle-\frac14 \scalprod{f}{\coth(\beta|k|/2)f}\right\}.
\end{equation}
The Weyl algebra is represented on ${\mathcal F}_\beta$ as $W(f)\mapsto W(f_\beta)$, for functions $f\in L^2({\mathbb R}^3)$ such that $ \scalprod{f}{\coth(\beta|k|/2)f}<\infty$. We denote the von Neumann algebra of the represented Weyl operators by ${\mathcal W}_\beta$.

The combined system-reservoir Hilbert space is then $\mathcal H$, \eqref{1}, and the von Neumann algebra of observables is 
\begin{equation}
\label{vna}
{\frak M}={\mathcal B}({\mathfrak H}_\S)\otimes\one_\S\otimes{\mathcal W}_\beta \subset {\mathcal B}({\mathcal H}).
\end{equation}
The coupled dynamics is given by
\begin{equation}
\label{dyn}
\alpha^t(A) =\e^{\i tL}A\e^{- \i tL}, \qquad A\in\frak M.
\end{equation}
It is generated by the self-adjoint Liouville operator acting on $\mathcal H$,
\begin{eqnarray}
L &=& L_0+\d I\label{2.4}\\
L_0 &=& L_\S+L_\R, \label{2.4'}\\
I &=& V - JVJ.\label{2.4''}
\end{eqnarray}
Here,  $L_\S=H_\S\otimes \bbbone_\S-\bbbone_\S\otimes H_\S$ and $H_\S$ is the system Hamiltonian acting on ${\mathfrak H}_\S$. $L_\r=d\Gamma(u)$ is the second quantization of multiplication by the radial variable $u$. The interaction $I$ in \eqref{2.4''} is ``in standard form", involving a self-adjoint interaction operator $V$ acting on $\mathcal H$ and the modular conjugation $J$, which acts as
\begin{equation}
J( A\otimes\bbbone_\S\otimes W(f_\beta(u,\Sigma)) )J = \bbbone_\S\otimes\overline A\otimes W(\overline f_\beta(-u,\Sigma)),
\label{2.6}
\end{equation}
where $\overline A$ is the matrix obtained from $A$ by taking entrywise complex conjugation (matrices are represented in the eigenbasis of $H_\S$). Note that by \eqref{2.3}, we have $\overline f_\beta(-u,\Sigma) = -e^{-\beta u/2}f_\beta (u,\Sigma)$. By the Tomita-Takesaki theorem \cite{BR}, conjugation by $J$ maps the von Neumann algebra of observables \eqref{vna} into its commutant. In particular, $V$ and $JVJ$ commute (strongly on a suitable domain). For more detail about this well-known setup we refer to \cite{JP,BFSrte,MSB2} and references therein. We have in mind two commonly used forms for $V$, 
\begin{equation}
\label{V}
V_1 = G\otimes\one_\S\otimes \phi(h_\beta)+\textrm{h.c.}\qquad \mbox{or} \qquad V_2=G\otimes\one_\S\otimes W(h_\beta)+\textrm{h.c.},
\end{equation} 
for some matrix $G$ on ${\mathfrak H}_\S$ and where $h_\beta$ is a (represented) {\em form factor}, obtained from an $h\in L^2({\mathbb R}^3)$ by \eqref{2.3}. The interaction $V_1$ is standard. $V_2$ comes about when considering the spin-boson system at arbitrary coupling strength \cite{KoMeSo}, see Section \ref{sbsect}.

The vector representing the uncoupled  $(\alpha^t,\beta)$-KMS state ($\d=0$) on $\mathfrak M$ is
\begin{equation}
\Omega_{0,\rm KMS} =  \Omega_{\S,\beta}\otimes \Omega,
\label{2.13}
\end{equation}
where $\Omega_{\S,\beta}$ is given in \eqref{systemkms}. 

In this setup of open systems, one can derive Condition (A2) from a global limiting absorption principle as follows.
\begin{thm}
\label{lapos}
Let $P_\R=\one_\S\otimes|\Omega\rangle\langle\Omega|$. 
Suppose that there is a dense set ${\mathcal D}\subset{\mathcal H}$ with ${\rm Ran} IP_\R \subset {\mathcal D}$, s.t. for all $\varphi,\psi\in\mathcal D$,
\begin{equation}
\label{nlapos1}
\sup_{z\in{\mathbb C}_-} |\tfrac{d^k}{dz^k}\scalprod{\phi}{R_z^{P_\R}\psi}|\le C(\phi,\psi),\quad k=0,\ldots,3.
\end{equation}
Then Condition (A2) holds with $\alpha=g/2=\min_{e\neq e'}\{|e-e'|\}/2$. 
\end{thm}
We prove Theorem \ref{lapos} in Appendix \ref{appendix1}. (The constant $C(\phi,\psi)$ in \eqref{nlapos1} may differ from that in (A2).)

\subsection{The spin-boson model for arbitrary coupling strength} 
\label{sbsect}
The spin-boson Hamiltonian is \cite{Leggett}
\begin{equation}
\label{1.3}
H=-\tfrac12\d\sigma_x+\tfrac12\varepsilon\sigma_z+H_\R+\tfrac12q_0\sigma_z\otimes\phi(h),
\end{equation}
where $\sigma_x$ and $\sigma_z$ are the Pauli matrices
$$
\sigma_x=
\begin{pmatrix}
0 & 1\\
1 & 0
\end{pmatrix}, \qquad
\sigma_z=
\begin{pmatrix}
1 & 0\\
0 & -1
\end{pmatrix} 
$$
and $\d, \varepsilon\in\mathbb R$ are the `tunneling matrix element' and the `detuning parameter', respectively. (We use units so that $\hbar$ takes the value one.) The reservoir Hamiltonian is $H_\R=\int_{{\mathbb R}^3} |k| a^*(k)a(k) d^3k$. The coupling constant is $q_0\in\mathbb R$. Associated to the Hamiltonian $H$ is the Liouvillean 
\begin{equation}
\label{sbliouv}
L = L_\S+L_\R+q_0I
\end{equation}
with 
\begin{eqnarray}
L_\s&=& H_\S\otimes\one_{{\mathbb C}^2} -\one_{{\mathbb C}^2}\otimes H_\S,\label{-=1}\\
L_\R&=&d\Gamma(u)\label{-=2}
\end{eqnarray}
where $H_\S=-\tfrac12\d\sigma_x+\tfrac12\varepsilon\sigma_z$, and $I$ given in \eqref{2.4''} with $V=\tfrac12 \sigma_z\otimes\phi(h)$. The total Hilbert space is given by \eqref{1}, \eqref{particleGNS}  with ${\mathfrak H}_\S={\mathbb C}^2$ and \eqref{2.1}.

In order to be able to analyze the spectrum of $L$ for arbitrarily large couplings $q_0\in\mathbb R$, one applies the unitary (`polaron'-) transformation $U$ (see \cite{KoMeSo})
\begin{equation}
\label{unitary}
U = u\, JuJ,\qquad \mbox{where}\qquad  u=\e^{\i\sigma_z\otimes\one_\S\otimes\phi(f_\beta)},
\end{equation}
resulting in a new Liouvillean \cite{Leggett,KoMeSo}
\begin{equation}
\label{nn1}
{\mathcal L} = ULU^* = {\mathcal L}_0+\d I,
\end{equation}
where
\begin{equation}
\label{nn2}
{\mathcal L}_0={\mathcal L}_\S+{\mathcal L}_\R = \tfrac{\varepsilon}{2}\big( \sigma_z\otimes\one_\S - \one_\S\otimes\sigma_z\big) +L_\R 
\end{equation}
and 
\begin{equation}
\label{nn3}
I = -\tfrac12 \big({\mathcal V}-J{\mathcal V}J\big),\qquad {\mathcal V}=\sigma_+\otimes\one_\S\otimes W(2f_\beta) + \sigma_-\otimes\one_\S\otimes W(-2f_\beta).
\end{equation}
Here, $\sigma_+$ and $\sigma_-$ are the raising and lowering operators and
\begin{equation}
\label{2.9'''}
f_\beta=(-\tfrac{i}{2}q_0h/u)_\beta.
\end{equation} 
The non-interacting KMS state associated to $\L_0$ is
\begin{equation}
\label{nonintkms}
\Psi_{0,{\rm KMS}} = \Psi_{\S,\beta}\otimes\Omega,\qquad \Psi_{\S,\beta}=\frac{\e^{-\beta\varepsilon/4}\varphi_+\otimes\varphi_+ +\e^{\beta\varepsilon/4}\varphi_-\otimes\varphi_-}{\sqrt{\e^{-\beta\varepsilon/2}+\e^{\beta\varepsilon/2}}}
\end{equation}
and the interacting KMS state associated to $\L$ is (Araki's perturbation theory of KMS states)
\begin{equation}
\label{intkms}
\Psi_{\rm KMS} = \frac{\e^{-\beta(\L_0+\d{\mathcal V})/2}\Psi_{0,{\rm KMS}}}{\|\e^{-\beta(\L_0+\d{\mathcal V})/2}\Psi_{0,{\rm KMS}} \|}.
\end{equation}
Note that the spectrum of ${\mathcal L}_0$ consists of a purely absolutely continuous part covering all of $\mathbb R$, in which are embedded the eigenvalues $e=\pm\varepsilon$ (each simple) and the doubly degenerate eigenvalue $e=0$.  The following is the main result of \cite{KoMeSo}:
\begin{thm}[\cite{KoMeSo}]
\label{KoMeSothm}
Recall that $(\cdot)_\beta$ is defined in \eqref{2.3}. Assume that $(1+|\i\partial_u|^\eta)(\i h/u)_\beta \in L^2({\mathbb R}\times S^2, du\times d\Sigma)$ for some $\eta>2$. Then, for any $q_0\in\mathbb R$, $q_0\neq 0$, there is a constant $\d_0>0$ s.t. if $0<|\d|\le\d_0$, then ${\mathcal L}$ has purely absolutely continuous spectrum covering $\mathbb R$ and a single, simple eigenvalue at zero. The associated eigenvector is $\Psi_{\rm KMS}$, \eqref{intkms}.
\end{thm}

{\em Remarks\ } 1. The eigenvalues $e=\pm\varepsilon$ of ${\mathcal L}_0$ are unstable, while $e=0$ is partially stable. The associated (simple) eigenvalue of $\mathcal L$ is $E=0$. 

2. The spectral properties of $L$ and $\mathcal L$ are the same, as the operators are unitarily equivalent to each other. 

3. The KMS state associated to $L$ is given by 
\begin{equation}
\label{equilstate}
\Omega_{\rm KMS} = U^*\Psi_{\rm KMS} =  \frac{\e^{-\beta(L_0+q_0\sigma_z\otimes\one_{{\mathbb C}^2}\otimes \phi(h_\beta))/2}\Omega_{0,{\rm KMS}}}{\| \e^{-\beta(L_0+q_0\sigma_z\otimes\one_{{\mathbb C}^2}\otimes \phi(h_\beta))/2}\Omega_{0,{\rm KMS}}\|},
\end{equation}
where $L_0=L_\S(\d)+L_\R$, see \eqref{-=1}, \eqref{-=2}. One shows that $\Omega_{0,\rm KMS}$ is in the domain of $\e^{-\beta(L_0+q_0\sigma_z\otimes\one_{{\mathbb C}^2}\otimes \phi(h_\beta))/2}$ for any $q_0$, $\d\in\mathbb R$ (see e.g. \cite{DJP,BFSrte,BR}).
\medskip

We now verify assumptions (A1)-(A4) for the spin-boson system, i.e., for the operator $\mathcal L$, \eqref{nn1}. 
The eigenprojections of ${\mathcal L}_0$ are given, for $e\in{\rm spec}({\mathcal L}_S)$,
by 
$$
P_e=\one[{\mathcal L}_\S=e]\otimes |\Omega\rangle\langle\Omega|.
$$
Since $\one[{\mathcal L}_\S=e]\big(\sigma_\pm\otimes\one_\S\big) \one[{\mathcal L}_\S=e]=0=\one[{\mathcal L}_\S=e]\big(\one_\S\otimes\sigma_\pm\big) \one[{\mathcal L}_\S=e]$, condition (A1) holds.

To verify the limiting absorption principle (A2), let $N=d\Gamma(\one)$ be the number operator  on Fock space \eqref{2.1} and put $\bar N=P^\perp_\R N$. Let $A=d\Gamma(\i\partial_u)$ and put $\bar A=P^\perp_\R A$. For $\alpha,\nu\ge 0$, define the norms
\begin{equation}
\label{alphanunorms}
\|\xi\|_{\alpha,\nu}= \|\bar{N}^{\nu/2}(1+\bar A^2)^{\alpha/2}\xi\|.
\end{equation}
We have the following regularity properties of the resolvent $R^{P_\R}_z$.
\begin{thm}
\label{MK-Lem1}
Let $\mu\ge 1$ and suppose that $\partial^j_uf_\beta\in L^2({\mathbb R}\times S^2,du\times d\Sigma)$, for $j=0,\ldots,2\mu+1$. We have
\begin{eqnarray}
\sup_{z\in{\mathbb C}_-} |\tfrac{d^{\mu-1}}{dz^{\mu-1}}\scalprod{\phi}{ R^{ P_\R}_z\psi}|&
\ple &  \|\phi\|_{\mu,2\mu}\,\|\psi\|_{\mu,2\mu}\label{eq-w17} \\
\sup_{z\in{\mathbb C}_-}|\tfrac{d}{d\d}\scalprod{\phi}{R^{ P_\R}_z\psi}|& \ple & \|\phi\|_{3,1}\|\psi\|_{3,1}.
\label{MK-eq19}
\end{eqnarray}
\end{thm}

We give a proof of Theorem \ref{MK-Lem1} in Appendix \ref{appendix1}. The bound \eqref{eq-w17} with $\mu=4$ implies \eqref{nlapos1}, with the dense set \begin{equation}
\label{deh}
{\mathcal D}=\{\psi\ :\ \|\psi\|_{4,8}<\infty\}.
\end{equation}
To see that ${\rm Ran}IP_\R\subset{\mathcal D}$, it suffices to check that 
\begin{equation}
\label{acheck}
\| N^4 (1+A^4)W(2f_\beta)\Omega\|<\infty.
\end{equation}
It is not hard to use the relation
\begin{equation}
\label{commut} {}[d\Gamma(D),W(f)]=W(f)(\phi(iDf)+\tfrac12\scalprod{f}{Df})
\end{equation}
(where $\phi$ is the field operator, see the proof of Lemma 3.1 in \cite{KoMeSo} for technical details) for $D=-\imath \partial_u$ and $D=\one$ to see that \eqref{acheck} holds provided $\partial_u^jf_\beta\in L^2({\mathbb R}\times S^2)$, $j=0,\ldots,4$.
 Therefore, Theorem \ref{MK-Lem1} combined with this last observation shows that the assumptions of Theorem \ref{lapos} are satisfied. Thus, by the latter theorem, assumption (A2) holds.

Next, assumption (A3) is shown to hold in Theorem \ref{lape}, \eqref{ddlocal}. Namely, the regularity in $\d$ of $\tfrac{d^k}{d\d^k} P_eIR^{P_e}_z IP_e$ is derived from that of $\tfrac{d^k}{d\d^k} P_eIR^{P_\R}_z IP_e$, given in Theorem \ref{MK-Lem1}, \eqref{MK-eq19}. 

The Fermi Golden Rule Assumption (A4) is verified by examining the level shift operators $\Lambda_0$ and $\Lambda_{\pm \varepsilon}$. $\Lambda_0$ is two-dimensional, given by (see \cite{KoMeSo}, Proposition 3.5)
\begin{equation}
\label{lambdanot}
\Lambda_0 = \i\tau^{-1} P^\perp_{\S,\beta}
\end{equation}
where $P^\perp_{\S,\beta}$ is the complement of   $P_{\S,\beta}=|\Omega_{\S,\beta}\rangle\langle\Omega_{\S,\beta}|$ in ${\rm Ran} \,\one[\L_\S=0]$ and where $\Omega_{\S,\beta}\propto \e^{-\beta \varepsilon/4}\varphi_+\otimes\varphi_++\e^{\beta\varepsilon/4}\varphi_-\otimes\varphi_-$  (see \eqref{systemkms} and \eqref{nn2}). Also,  
\begin{equation}
\tau^{-1} = \int_0^\infty d t\cos(\varepsilon t) \cos\left[\frac{q_0^2}{\pi}\, Q_1(t)\right]e^{-\frac{q_0^2}{\pi}\, Q_2(t)}
\label{t2}
\end{equation}
with
$$
Q_1(t) = \int_0^\infty d\omega \frac{J(\omega)}{\omega^2}\sin(\omega t)\quad \mbox{and}\quad 
Q_2(t) = \int_0^\infty  d\omega \frac{J(\omega)(1-\cos(\omega t))}{\omega^2} \coth(\beta\omega/2).
$$
Here, $J(\omega)$ is the {\em spectral density} of the reservoir, defined by 
\begin{equation}
J(\omega) = \textstyle\frac{\pi}{2}\omega^2\int_{S^2} |h(\omega,\Sigma)|^2 d\Sigma,\qquad \omega\geq 0,
\label{1.10}
\end{equation}
the integral being taken over the angular part in ${\mathbb R}^3$. The function $h$ is the form factor in the interaction \eqref{1.3}.\footnote{The spectral density is related to the Fourier transform of the reservoir correlation function $C(t)=\omega_{R,\beta}(e^{i tH_R}\varphi(h)e^{-i tH_R}\varphi(h))$ by $J(\omega)=\sqrt{\pi/2}\tanh(\beta\omega/2)[\widehat{C}(\omega)+\widehat{C}(-\omega)]$. Of course, it is assumed here, as it is in \cite{Leggett}, that the integral in \eqref{t2} does not vanish, so that $\tau<\infty$ is a finite relaxation time.} Relation \eqref{lambdanot} gives 
\begin{equation}
\label{lambdanotnot}
\lambda_{0,0}=0\qquad \mbox{and}\qquad \lambda_{0,1}=\i\tau^{-1}.
\end{equation}
Hence assumption (A4)(2) holds. The resonances $\lambda_{\pm\epsilon,0}$ are the eigenvalues of the one-dimensional level shift operators $\Lambda_{\pm\varepsilon}$, which are easily calculated to be
\begin{equation}
\label{lambdanotepsilon}
\lambda_{\pm\varepsilon,0} = \pm x+\tfrac12 \i\tau^{-1},
\end{equation}
where $\pm x$ is the real part. Assumption (A4)(1) thus holds.

 Set $\varphi_{+-}= \varphi_+\otimes\varphi_-$ etc. and
\begin{eqnarray*}
X_0&=&W(f_\beta) JW(f_\beta)\Omega,\qquad \ \, X_0^*=W(f_\beta)^*J W(f_\beta)^*\Omega,\\
X_+&=&W(f_\beta)^* JW(f_\beta)\Omega,\qquad X_-=JX_+=W(f_\beta)JW(f_\beta)^*\Omega.
\end{eqnarray*} 
The dynamics of the spin-boson system at arbitrary coupling is then explicitly given as follows.

\begin{cor}[Dynamics of the spin-boson system at arbitrary coupling strength]
\label{cor1}
Suppose that $uf_\beta, \partial_u^j f_\beta\in L^2({\mathbb R}\times S^2)$ for $j=0,\ldots,4$. For any $q_0\in\mathbb R$ there is a $\d_0>0$ s.t. if $0<|\d|<\d_0$ then the following holds. 

Denote by $\Pi_0=|\Omega_{\rm KMS}\rangle\langle\Omega_{\rm KMS}|$ the projection onto the coupled KMS state, \eqref{equilstate}. Let $\phi,\psi\in\dom(L_\R)\cap \dom(N^{17/2}+1)(A^4+1)$. We have for all $t>0$
\begin{eqnarray}
\scalprod{\phi}{\e^{\i tL}\psi} &=&  \scalprod{\phi}{\Pi_{0}\psi} +  \e^{\i t \d^2 a_0}\scalprod{\phi}{\Pi'_0\psi}\label{dynthm1}\\
&&  + \e^{\i t (\varepsilon+ \d^2 a_\varepsilon)}\scalprod{\phi}{\Pi'_\varepsilon\psi}  + \e^{\i t (-\varepsilon+ \d^2 a_{-\varepsilon})}\scalprod{\phi}{\Pi'_{-\varepsilon}\psi}  +R(t),
\nonumber
\end{eqnarray}
where $R(t) \ple \frac{1}{t}$ and 
\begin{equation*}
a_0 = \i\tau^{-1} +O(\d), \qquad a_{\pm\varepsilon} =\pm x+\tfrac{\i}{2}\tau^{-1} +O(\d), 
\end{equation*}
\begin{eqnarray}
\Pi'_0 &=& (\e^{\beta\varepsilon/2}+\e^{-\beta\varepsilon/2})^{-1}\big( \e^{\beta\varepsilon/2}|\varphi_{++}\rangle\langle\varphi_{++}|\otimes|X_0\rangle\langle X_0| -|\varphi_{++}\rangle\langle\varphi_{--}|\otimes|X_0\rangle\langle X^*_0|\nonumber\\ &&-|\varphi_{--}\rangle\langle\varphi_{++}|\otimes|X_0^*\rangle\langle X_0| +\e^{-\beta\varepsilon/2} |\varphi_{--}\rangle\langle\varphi_{--}|\otimes|X_0^*\rangle\langle X_0^*|   \big) +O(\d),\label{--1}\\
\Pi'_{\varepsilon} &=& |\varphi_{+-}\rangle\langle\varphi_{+-}|\otimes |X_+\rangle\langle X_+|+O(\d),\label{--2}\\
{} \qquad \Pi'_{-\varepsilon} &=&|\varphi_{-+}\rangle\langle\varphi_{-+}|\otimes |X_-\rangle\langle X_-|+O(\d).\label{--3}
\end{eqnarray}
\end{cor}

Note that the dynamics in Corollary \ref{cor1} is expressed with respect to the original Liouville operator $L$ (not the unitarily equivalent $\L$).

{\em Proof of Corollary \ref{cor1}. } We have $\scalprod{\phi}{\e^{\i tL}\psi}=\scalprod{U\phi}{\e^{\i t\L}U\psi}$.  Theorem \ref{nthm1} gives the resonance expansion for $\scalprod{U\phi}{\e^{\i t\L}U\psi}$ provided $U\phi, U\psi\in{\mathcal D}$ s.t. $\L U\phi,\L U\psi\in{\mathcal D}$, where $\mathcal D$ is the set of vectors with finite $\|\cdot\|_{4,8}$ norm (see after Theorem \ref{MK-Lem1}).One can easily show the bound 
\begin{eqnarray}
\nonumber
\|U\psi\|_{2\alpha,\nu}&\ple& \|\psi\|_{2\alpha,\nu}.
\label{UD}
\end{eqnarray}
 We conclude from \eqref{UD} that if $\psi\in \dom(N^{\nu+1/2}+1)(A^{2\alpha}+1)$ for some $\nu\ge 1/2$ and $\alpha\ge0$, and if $(\i\partial_u)^jf_\beta\in L^2({\mathbb R}\times S^2)$ for $j=0,\ldots,2\alpha$, then $\|U\psi\|_{2\alpha,\nu}<\infty$.

We infer that if $\psi\in\dom(L_\R)\cap  \dom(N^{17/2}+1)(A^4+1)$, then $U\psi\in{\mathcal D}$ and $\L U\psi\in{\mathcal D}$. (To see the latter inclusion, we use \eqref{commut} with $D=u$ and proceed as in \eqref{UD}.)

The operators $\Pi'_0$ and  $\Pi'_{\pm\varepsilon}$ are given by 
\begin{equation}
\label{projU}
\Pi'_0= U^*P^\perp_{\S,\beta}P_\R U +O(\d),\qquad  \Pi'_{\pm\varepsilon}=U^* P_{\pm\varepsilon}U +O(\d),
\end{equation}
where $P_{\pm\varepsilon}$ are the eigenprojections of $\L_0$ associated to the eigenvalues $\pm\varepsilon$,  $P^\perp_{\S,\beta}$ is defined after \eqref{lambdanot} and $P_\R=\one_\S\otimes|\Omega\rangle\langle\Omega|$. It is easy to calculate  $U\varphi_{+-}\otimes\Omega = \varphi_{+-}\otimes W(f_\beta)^* JW(f_\beta)J\Omega$, which shows \eqref{--2}. Similarly one obtains \eqref{--3}. For \eqref{--1}, we first calculate
$$
P^\perp_{\S,\beta} = (\e^{\beta\varepsilon/2}+\e^{-\beta\varepsilon/2})^{-1}\big( \e^{\beta\varepsilon/2}|\varphi_{++}\rangle\langle\varphi_{++}| -|\varphi_{++}\rangle\langle\varphi_{--}|-|\varphi_{--}\rangle\langle\varphi_{++}| +\e^{-\beta\varepsilon/2}|\varphi_{--}\rangle\langle\varphi_{--}|\big).
$$
Then we use \eqref{projU} to arrive at \eqref{--1}. This completes the proof of Corollary \ref{cor1}. \hfill \qed

\bigskip

Being a KMS state, $\Omega_{\rm KMS}$ given in \eqref{equilstate} is {\em separating} for $\mathfrak M$, which means that ${\mathfrak M}'\Omega_{\rm KMS}$ is dense in $\mathcal H$, where ${\mathfrak M}'$ is the commutant of $\mathfrak M$, see \cite{BR}. Any (normal) state $\omega$ on $\mathfrak M$ is given by a normalized vector $\Psi\in\mathcal H$ via  $\omega(A)=\scalprod{\Psi}{A\Psi}$. We introduce the dense set
\begin{equation}
\label{dehnot}
{\mathcal D}_0 =\dom(L_\R)\cap \dom(N^{17/2}+1)(A^4+1).
\end{equation}
The set of states $\omega$ arising from vectors in
\begin{equation}
\label{regularstates}
\big\{\Psi\in{\mathcal H}\ :\ \|\Psi\|=1, \, \Psi = B\Omega_{\rm KMS} \mbox{\ for some $B\in{\mathfrak M}'$, s.t.\  } B^*\Psi\in{\mathcal D}_0\big\}
\end{equation}
is dense (in the norm of states on $\mathfrak M$). We call it the set of regular states, ${\mathcal S}_{\rm reg}$. We also introduce the regular observables,
\begin{equation}
\label{regobs}
{\mathfrak M}_{\rm reg} = \big\{ A\in{\mathfrak M}\ :\ A\Omega_{\rm KMS}\in{\mathcal D}_0\big\}.
\end{equation}
Let us denote the coupled equilibrium state by
\begin{equation}
\label{jointKMSstate}
\omega_{\rm KMS}(A) = \scalprod{\Omega_{\rm KMS}}{A\Omega_{\rm KMS}}.
\end{equation}

\begin{cor}[Return to equilibrium] 
\label{cor2}
 For any $\omega_0\in{\mathcal S}_{\rm reg}$, $A\in{\mathfrak M}_{\rm reg}$, $t\ge 0$, we have 
\begin{equation}
\label{mm63}
\left| \omega_0(\alpha^t(A)) - \omega_{\rm KMS}(A)\right| \le C_{A,\omega_0} \left[ \e^{-\d^2t/2\tau} + 1/t \right].
\end{equation}
The constant $C_{A,\omega_0}$ depends on the initial state $\omega_0$ and the observable $A$, but not on $t,\d$. 
\end{cor}

One readily verifies that  all states of the form $\omega_\S\otimes\omega_\R$, where $\omega_\S$ is arbitrary and $\omega_\R$ is the reservoir equilibrium, belong to ${\mathcal S}_{\rm reg}$. Moreover, all observables on the system alone belong to ${\mathfrak M}_{\rm reg}$, since $\Omega_{\rm KMS}\in{\mathcal D}_0$. 


\bigskip

{\em Proof of Corollary \ref{cor2}.\ } Let $\Psi=B\Omega_{\rm KMS}$, $B\in{\mathfrak M}'$, be the vector representing $\omega_0$. Since $B$ commutes with $\alpha^t(A)$ and $\Omega_{\rm KMS}$ is in the kernel of $L$, we have $\alpha^t(A)\Psi=B\alpha^t(A)\Omega_{\rm KMS} = B\e^{\i tL}A\Omega_{\rm KMS}$. Thus,
\begin{equation}
\label{p1}
\omega_0(\alpha^t(A)) = \scalprod{\Psi}{\e^{\i tL} A\e^{-\i tL}\Psi} = \scalprod{B^*\Psi}{\e^{\i tL}A\Omega_{\rm KMS}}.
\end{equation}
Now we apply \eqref{dynthm1} and, using that $\Pi_0=|\Omega_{\rm KMS}\rangle\langle\Omega_{\rm KMS}|$, obtain directly \eqref{mm63}.\hfill \qed

\section{Proof of Theorem \ref{nthm1}}
\label{proofthmnthm1sect}

{}For $z\in{\mathbb C}_-$, $|{\rm Re}z-e|\le \alpha$ we define the operator
\begin{equation}
\label{Az}
A_z=-P_eI R_z^{P_e} IP_e.
\end{equation}
As $z\rightarrow e$ and $\d\rightarrow 0$, $A_e$ approaches the level shift operator $\Lambda_e$. More precisely, we have the following result, in which $\delta$ is, recall, given by \eqref{delta}. 
\begin{lem}
\label{lemma1}
There is a constant $c$ such that if $|\d|$, $|{\rm Re}z-e| <c \delta$, and $z\in{\mathbb C}_-$, then

1.  All eigenvalues of $A_z$ are distinct. Call them $a_{e,j}=a_{e,j}(z)$, $j=0,\ldots,m_e-1$. Each $a_{e,j}$ satisfies $|\lambda_{e,j}-a_{e,j}|<\delta/2$ for exactly one eigenvalue $\lambda_{e,j}$ of $\Lambda_e$.

2. The eigenvalues $a_{e,j}(z)$ of $A_z$, and the associated Riesz projections $Q_j(z)$ are analytic in $z\in{\mathbb C_-}$, $|{\rm Re}z-e|<c\delta$ and continuous as ${\rm Im}z\rightarrow 0_-$. They satisfy the bounds $\tfrac{d^k}{dz^k}a_{e,j}(z)$, $\tfrac{d^k}{dz^k}Q_j(z)\ple 1$ for $k=0,\ldots,3$, uniformly for $|{\rm Re}z-e|<c\delta$ and ${\rm Im} z \le 0$.
\end{lem}
The simplicity of the spectrum implies the spectral representation
\begin{equation}
\label{naz}
A_z = \sum_{j=0}^{m_e-1} a_{e,j}(z) Q_j(z). 
\end{equation}

{\em Proof of Lemma \ref{lemma1}.\ } When necessary, we display the $\d$-dependence of $A_z$ by $A_z(\d)$. (Here, $A_x$ for $x\in\mathbb R$ is understood as the limit of $A_z$, as $z\rightarrow x$, $z\in{\mathbb C}_-$.) We have $A_e(0)=\Lambda_e$ and, by assumption (A3),
\begin{equation}
\label{04}
\|\Lambda_e-A_z(\d)\|\ple |\d|+|z-e|.
\end{equation}
Assumption (A4) implies that 
\begin{equation}
\label{03}
\|(\Lambda_e-\zeta)^{-1}\| \ple \frac{1}{{\rm dist}(\zeta,{\rm spec}(\Lambda_e))}.
\end{equation} 
Using the standard Neumann series for resolvents, together with the estimates \eqref{04} and \eqref{03}, yields
\begin{equation}
\label{05}
(A_z(\d)-\zeta)^{-1} = (\Lambda_e-\zeta)^{-1} +O\left(\frac{|\d|+|z-e|}{{\rm dist}(\zeta,{\rm spec}(\Lambda_e))^2}\right),
\end{equation}
provided $|\d|$, $|z-e| < c_0\, {\rm dist}(\zeta,{\rm spec}(\Lambda_e))$, for some constant $c_0$ independent of $\d$, $z$. Let $\C_j$ be the circle centered at $\lambda_{e,j}$ with radius $\delta/2$ and define
\begin{equation}
\label{06}
Q_j(\d,z) = \frac{-1}{2\pi \i}\oint_{\C_j}(A_z(\d)-\zeta)^{-1}d \zeta.
\end{equation}
Note that $Q_j(0,e)$ is the Riesz eigenprojection of $\Lambda_e$ associated to the eigenvalue $\lambda_{e,j}$. Using \eqref{05} and \eqref{06} we obtain that $\|Q_j(\d,z)-Q_j(0,e)\|<1$, provided $|\d|$, $|z-e|<c_0\delta/2$ and $|\d|+|z-e| < c_1\delta^2/4$, for some $c_1$ independent of $\d$, $z$. This proves point 1. of Lemma \ref{lemma1}.

Next, 
\begin{equation}
\label{07}
Q_j'=\tfrac{d}{dz}Q_j=\frac{1}{2\pi\i}\oint_{\C_j}(A_z(\d)-\zeta)^{-1}A'_z(\d) (A_z(\d)-\zeta)^{-1}d \zeta \ \ple \,1.
\end{equation}
Since $a_{e,j}(z)={\rm Tr} A_zQ_j(z)$ we get $a'_{e,j}(z)={\rm Tr}(A'_zQ_j(z)+A_zQ'_j(z))\ple 1$. The statements about the higher derivatives follow in the same manner. This shows point 2 and completes the proof of Lemma \ref{lemma1}.\hfill \qed

\subsection{Estimates for $z$ in a vicinity of a partially stable eigenvalue $e$.}
\label{nestzero}

We introduce the operators
\begin{equation}
\label{qtilde}
\widetilde Q_j = -Q_j(e)P_eIR^{P_e}_{e-\i 0_+}-R^{P_e}_{e-\i 0_+}IP_eQ_j(e) +\d R^{P_e}_{e-\i 0_+} IP_eQ_j(e)P_eIR^{P_e}_{e-\i 0_+},\qquad j\ge 1,
\end{equation}
where $Q_j(e)$ is determined by \eqref{Az} and $R^{P_e}_{e-\i 0_+}$ is the limit of $R^{P_e}_{z}=(\bar L-z)^{-1}$ as $z$ approaches $e$ through the lower half plane. Note that $Q_j(e)=P_{e,j}+O(|\d|)$, see also \eqref{n7}. 

\begin{prop}
\label{propzero}
Let $\Pi_{E_e}$ be the orthogonal projection associated to the eigenvalue $E_e$ of $L$. We have, for $\phi,\psi\in\mathcal D$,
\begin{equation}
\label{13.01}
J_e(t) = \e^{\i t E_e} \scalprod{\phi}{\Pi_{E_e}\psi} + \sum_{j=1}^{m_e-1} \e^{\i t(e+\d^2  a_{e,j}(e))}\scalprod{\phi}{(Q_j(e)+\d\, \widetilde Q_j)\psi} +R_e(t) 
\end{equation}
where (recall that $w>0$ is the arbitrary parameter in \eqref{14})
$$
|R_e(t)|\le C \,\frac{1+\e^{wt}/t}{t},
$$
for a constant $C$ independent of $\d,t,w$, and where 
\begin{equation}
\label{13.1}
a_{e,j}(e) = \lambda_{e,j} +O(|\d|),\qquad Q_j(e)=P_{e,j}+O(|\d|).
\end{equation}
\end{prop}
{\em Proof of Proposition \ref{propzero}.\ } We apply the Feshbach map with projection $P_e$ (having rank $m_e$), 
\begin{equation}
\label{01}
\F(z) \equiv\F(L-z;P_e)= e-z+\d^2A_z,
\end{equation}
where $A_z$ is given in \eqref{Az}. Due to Condition (A2), $z\mapsto A_z$ is analytic for $z\in{\mathbb C}_-$, $|z-e|<\alpha$, and its $z$-derivatives up to degree 3 stay bounded as ${\rm Im}z\rightarrow 0_-$. According to the decomposition \eqref{0}, we have
\begin{equation}
\label{08}
\int_{\G_e}\e^{\i tz}\scalprod{\phi}{\textstyle\frac{d}{dz}R_z\psi} dz = \int_{\G_e}\e^{\i tz}\left[ \scalprod{\phi}{\textstyle\frac{d}{dz}\F(z)^{-1}\psi} + \scalprod{\phi}{\textstyle\frac{d}{dz}{\mathcal B}(z)\psi}+ \scalprod{\phi}{\textstyle\frac{d}{dz}R^{P_e}_z\psi}  \right] dz,
\end{equation}
where 
\begin{equation}
{\mathcal B}(z) = -\d \F(z)^{-1} P_e I R_z^{P_e} - \d R_z^{P_e} I P_e\F(z)^{-1} +\d^2 R_z^{P_e} I P_e \F(z)^{-1} P_e I R_z^{P_e} .
\label{09'}
\end{equation}
To examine $\F(z)^{-1}$, we use the spectral representation of the operator $A_z$.

\subsubsection{The contribution to \eqref{08} from $\tfrac{d}{dz}\F(z)^{-1}$} In this subsection, we will simply write 
$$
a_j \equiv a_{e,j},\qquad 0\le j\le m_e-1,
$$
to ease the notation. Due to \eqref{01} and Lemma \ref{lemma1}, 
\begin{equation}
\label{033}
\F(z)^{-1} = \sum_{j=0}^{m_e-1}\frac{1}{e-z+\d^2a_j}Q_j.
\end{equation}
We analyze the first term on the right side of \eqref{08}, using that  
\begin{equation}
\tfrac{d}{dz}\F(z)^{-1} =\sum_{j=0}^{m_e-1} T_j,\qquad \mbox{where}\qquad T_j=\frac{1-\d^2a_j'}{(z-e-\d^2a_j)^2}Q_j + \frac{1}{e-z+\d^2a_j}Q'_j .
\label{09}
\end{equation}
We examine the singularities of $T_j$ in $z$. By the isospectrality of the Feshbach map, we know that $e-E_e+\d^2A_{E_e}$ has an eigenvalue zero (see also \eqref{n5}). Therefore, $e-E_e+\d^2a_0(E_e)=0$. Also, $a_j(z)=\lambda_{e,j}+O(|\d|+|z-e|)$ for $j=1,\ldots,m_e-1$. Consider first $T_0$. We have 
$$
z-e-\d^2a_0(z)=z-e-\d^2a_0(E_e)+\d^2(a_0(E_e)-a_0(z))=z-E_e+\d^2(a_0(E_e)-a_0(z))
$$
and so
\begin{eqnarray}
\label{012}
\frac{1-\d^2a_0'}{(z-e-\d^2a_0)^2} &=&\frac{1}{(z-E_e)^2}\frac{1-\d^2a'_0(z)}{[1-\d^2(a_0(z)-a_0(E_e))/(z-E_e)]^2}\\
& =& \frac{1}{(z-E_e)^2}\frac{1}{1-\d^2 a'_0(E_e)} +\frac{h(z)}{(z-E_e)^2},
\nonumber
\end{eqnarray}
where
\begin{equation}
\label{013}
h(z) = \frac{1-\d^2a'_0(z)}{[1-\d^2(a_0(z)-a_0(E_e))/(z-E_e)]^2} -\frac{1-\d^2a'_0(E_e)}{[1-\d^2a_0'(E_e)]^2}  =O(\d^2|z-E_e|^2).
\end{equation}
To arrive at \eqref{013}, we expand $h(z)$ around a point $z_0\in {\mathbb C}_-$ which is very close to $E_e$,
$$
h(z)=h(z_0) +(z-z_0)h'(z_0)+\int_{[z_0,z]}d s\int_{[s,z_0]}ds'h''(s').
$$
The integrals are over paths (straight lines) in the lower complex plane. Then, sending $z_0\rightarrow E_e$, using that $h(E_e)=h'(E_e)=0$ and controlling the double integral with the third derivative of $h$, we arrive at \eqref{013}. In this argument, we assume the derivatives up to order three to have a continuous extension as ${\rm Im} z\rightarrow 0_-$.

Thus
\begin{eqnarray}
\label{014}
T_0 &=& \frac{1}{(z-E_e)^2}\frac{Q_0(E_e)}{1-\d^2 a'_0(E_e)} +  \frac{1}{z-E_e}\frac{[Q_0(z)-Q_0(E_e)]/(z-E_e)}{1-\d^2 a'_0(E_e)}\\
&&-\frac{1}{z-E_e}\frac{Q'_0(z)}{1-\d^2[a_0(z)-a_0(E_e)]/(z-E_e)} +O(\d^2).
\nonumber
\end{eqnarray}
An expansion of the {\em sum} of  the second and third term on the right side of \eqref{014} shows that this term is $O(1)$ uniformly in $z\in\G_e$, giving the bound
\begin{equation}
\label{015}
T_0= \frac{1}{(z-E_e)^2}\frac{Q_0(E_e)}{1-\d^2 a'_0(E_e)} +O(1+\d^2).
\end{equation}
Therefore, 
\begin{equation}
\label{016}
\int_{\G_e}\e^{\i tz}\scalprod{\phi}{T_0\psi}=\frac{\scalprod{\phi}{Q_0(E_e)\psi}}{1-\d^2 a'_0(E_e)} \int_{\G_e} \frac{\e^{\i tz}}{(z-E_e)^2}dz + O(1).
\end{equation}
The remaining integral on the right side is estimated using the standard Cauchy formula from complex analysis. Namely, we complete $\G_e$ into a rectangular closed path, adding the vertical pieces ${\mathcal C}_\pm=\{e\pm\alpha+\i y\ :\ y\in [-w,R]\}$ and the horizontal roof $\{x+\i R\ :\ e-\alpha\le x\le e+\alpha\}$.  Then we obtain from the Cauchy integral formula of basic complex analysis, upon taking $R\rightarrow\infty$, that 
\begin{equation}
\label{017}
\int_{\G_e}\frac{\e^{\i tz}}{(z-E_e)^2}dz = 2\pi\i (\e^{\i tz})'|_{z=E_e} + O(\e^{wt}/t) = \i t\, 2\pi\i\, \e^{\i tE_e}  + O(\e^{wt}/t).
\end{equation}
The remainder term comes from the integrals along the two vertical pieces of the path, which are bounded above by $\int_{-w}^\infty \e^{-yt}dy$. Note that $w>0$ is arbitrary (see \eqref{4}) and we will take $w\rightarrow 0$ which will make the remainder in \eqref{017} to be $O(1/t)$. Combining \eqref{016} and \eqref{017} yields
\begin{equation}
\label{018}
\frac{1}{\i t}\frac{1}{2\pi\i} \int_{\G_e}\e^{\i tz}\scalprod{\phi}{T_0\psi}dz = \e^{\i tE_e} \frac{ \scalprod{\phi}{Q_0(E_e)\psi}}{1-\d^2 a'_0(E_e)}+O(1/t+\e^{wt}/t^2).
\end{equation}

Next, we analyze $T_j$ in \eqref{09}, for  $j\ge1$.  Recall that $a_j(e)|_{\d=0}=\lambda_{e,j}$ are the eigenvalues with strictly positive imaginary part of the level shift operator $\Lambda_e$. The integrands behave in a  different way now, since ${\rm Im}a_j(e)>0$ (while before, $a_0(E_e)=0$). The following bound is useful,
\begin{eqnarray}
|z-e-\d^2a_j(z)| &=& \left| z-e-\d^2\lambda_{e,j} -\d^2(a_j(e) - \lambda_{e,j})-\d^2(a_j(z)-a_j(e)) \right|\nonumber\\
&\ge& |z-e-\d^2\lambda_{e,j}| -c_1\d^2(|\d|+|z-e|)\nonumber\\
&\ge& \tfrac12  |z-e-\d^2\lambda_{e,j}|,
\label{020}
\end{eqnarray}
provided that $|\d|$, $|z-e|\le c_2 {\rm Im}\lambda_{e,j}$, where $c_1$  and $c_2=1/(2c_1)$ are independent of $\d$ and $z$. To arrive at the last inequality, \eqref{020}, we proceed as follows: the inequality is equivalent to $2c_1\d^2(|\d|+|z-e|)\le |z-e-\d^2\lambda_{e,j}|$. Now  $|z-e-\d^2\lambda_{e,j}| \ge |{\rm Im}z-\d^2{\rm Im}\lambda_{e,j}|\ge \d^2{\rm Im}\lambda_{e,j}$, since ${\rm Im} z<0$ and ${\rm Im}\lambda_{e,j}>0$.

We have
\begin{eqnarray}
\label{019}
T_j &=& \frac{Q_j(e)}{(z-e-\d^2a_j(e))^2}  - \frac{\d^2a'_j(z)Q_j(z)}{(z-e-\d^2a_j(z))^2}+ R_j+S_j,
\end{eqnarray}
where 
\begin{eqnarray}
R_j &=& \frac{1}{z-e-\d^2a_j(z)}\left[\frac{z-e}{z-e-\d^2a_j(z)} \frac{Q_j(z)-Q_j(e)}{z-e}-Q_j'(z)\right] \nonumber\\
&=& \frac{1}{z-e-\d^2a_j(z)}\left[\frac{Q_j(z)-Q_j(e)}{z-e}-Q_j'(z) +\frac{\d^2a_j(z)}{z-e-\d^2a_j(z)}\frac{Q_j(z)-Q_j(e)}{z-e} \right] \nonumber\\
&=& O\left(\frac{|z-e|}{|z-e-\d^2\lambda_{e,j}|} + \frac{\d^2}{|z-e-\d^2\lambda_{e,j}|^2}\right)
\label{021}
\end{eqnarray}
and
\begin{equation}
\label{022}
S_j = \frac{-2\d^2 (z-e)(a_j(e)-a_j(z)) +\d^4[a_j(e)^2-a_j(z)^2]}{ [z-e-\d^2a_j(z)]^2 [z-e-\d^2a_j(e)]^2} Q_j(0)= O\left( \frac{\d^2|z-e|^2+\d^4|z-e|}{|z-e-\d^2\lambda_{e,j}|^4} \right).
\end{equation}
To arrive at the estimates \eqref{021} and \eqref{022} we have used \eqref{020}. Similarly, the second term on the right side of \eqref{019} is $O(\d^2/|z-e-\d^2\lambda_{e,j}|^2)$ and so we obtain
\begin{equation}
\label{023}
T_j = \frac{Q_j(e)}{(z-e-\d^2a_j(e))^2}  +\widetilde T_j,
\end{equation}
with
\begin{equation}
\label{023.1}
\|\widetilde T_j\|\ple \frac{|z-e|}{|z-e-\d^2\lambda_{e,j}|} + \frac{\d^2}{|z-e-\d^2\lambda_{e,j}|^2}+ \frac{\d^2|z-e|^2+\d^4|z-e|}{|z-e-\d^2\lambda_{e,j}|^4}.
\end{equation}
Note that, with $\lambda_{e,j}=\xi_j+\i\eta_j$ and $w, \eta_j>0$, we have 
$$
\frac{|z-e|}{|z-e-\d^2\lambda_{e,j}|}\le 1+\frac{\d^2|\lambda_{e,j}|}{|z-e-\d^2\lambda_{e,j}|}\le1+\frac{|\lambda_{e,j}|}{\eta_j}\ple 1,
$$ 
so
\begin{equation}
\label{024}
\int_{\G_e}\frac{|z-e|}{|z-e-\d^2\lambda_{e,j}|}dx \ple 1.
\end{equation}
Also,
\begin{equation}
\label{025}
\int_{\G_e} \frac{\d^2}{|z-e-\d^2\lambda_{e,j}|^2} dx \le2\int_0^{\alpha+\d^2|\xi_j|}\frac{\d^2}{x^2 +\d^4\eta_j^2} dx\le 2\int_0^\infty\frac{dy}{y^2+ \eta_j^2}\ple 1.
\end{equation}
Very similarly, one sees that
\begin{equation}
\label{0.26}
\int_{\G_e}\frac{\d^2|z-e|^2+\d^4|z-e|}{|z-e-\d^2\lambda_{e,j}|^4}dx\ple 1
\end{equation}
and it follows that 
\begin{equation}
\label{027}
\int_{\G_e}\e^{\i tz} \widetilde T_j \,dz \ple 1.
\end{equation}
Next, using Cauchy's integral formula as above (\eqref{015}-\eqref{017}), we obtain that 
\begin{equation}
\label{028}
\int_{\G_e} \e^{\i tz}\frac{Q_j(e)}{(z-e-\d^2a_j(e))^2}dz = \i t\, 2\pi\i \, \e^{\i t (e+\d^2 a_j(e))} Q_j(e) +O(\e^{wt}/t).
\end{equation}
Combining \eqref{028} with \eqref{023} and \eqref{027}, we see that for all $j\geq 1$,
\begin{equation}
\label{029}
\frac{1}{\i t}\frac{1}{2\pi \i} \int_{\G_e} \e^{\i tz} \scalprod{\phi}{T_j\psi} dz = \e^{\i t(e+\d^2a_j(e))}\scalprod{\phi}{Q_j(e)\psi} +O(1/t+\e^{w t}/t^2).
\end{equation}

At this point, it is instructive to explain the coefficient $(1-\d^2a'_0(E_e))^{-1}$ in front of the non-decaying term in \eqref{018}. Let $\Pi_{E_e}$ be the projection onto the embedded eigenvalue $E_e$ of $L$. We have (in the strong sense) $\Pi_{E_e}=\lim_{\epsilon \rightarrow 0_+}(\i \epsilon) (L-E_e+\i \epsilon)^{-1}$, so by \eqref{0},
\begin{equation}
\label{031}
\Pi_{E_e}=\lim_{\epsilon \rightarrow 0_+}(\i \epsilon)\left\{  \F(L-E_e+\i \epsilon;P_e)^{-1} +{\mathcal B}(E_e-\i\epsilon) +R_{E_e-\i\epsilon}^{P_e}\right\}.
\end{equation}
The $P_e$-block of the decomposition is, by \eqref{033},
\begin{equation}
P_e\Pi_{E_e}P_e = \lim_{\epsilon\rightarrow 0_+}(\i\epsilon)\F(L-E_e+\i\epsilon;P_e)^{-1}
=\lim_{\epsilon\rightarrow 0_+}\sum_{j=0}^{m_e-1}\frac{\i\epsilon}{e-E_e+\i\epsilon+\d^2a_j(E_e-\i\epsilon)} Q_j(E_e-\i\epsilon).
\end{equation}
{}For $j\ge 1$, we have ${\rm Im}a_j(E_e)>0$ and the corresponding term in the sum vanishes in the limit $\epsilon\rightarrow 0$. Hence
\begin{equation}
\label{032}
P_e\Pi_{E_e}P_e=\lim_{\epsilon\rightarrow 0_+}(\i\epsilon) \frac{Q_0(E_e-\i\epsilon)}{e-E_e+\i\epsilon+\d^2a_0(E_e-\i\epsilon)}  =\frac{Q_0(E_e)}{1-\d^2a_0'(E_e)}.
\end{equation}
We have used the relation $e-E_e+\d^2a_0(E_e)=0$ (see after \eqref{09}). Therefore, the non-decaying, oscillating term on the right side in \eqref{018} is
\begin{equation}
\label{034}
 \frac{\e^{\i tE_e}}{1-\d^2a_0'(E_e)} \scalprod{\phi}{Q_0(E_e)\psi} =\e^{\i tE_e} \scalprod{\phi}{P_e\Pi_{E_e}P_e\psi}.
\end{equation}
Finally, we combine \eqref{018}, \eqref{029} and \eqref{034} to arrive at
\begin{eqnarray}
\label{035}
\frac{1}{\i t}\frac{1}{2\pi\i}\int_{\G_e}\e^{\i tz}\scalprod{\phi}{\tfrac{d}{dz}\F(z)^{-1}\psi} d z &=&\e^{\i tE_e} \scalprod{\phi}{P_e\Pi_{E_e}P_e \psi}+\sum_{j=1}^{m_e-1} \e^{\i t(e+\d^2a_j(e))}\scalprod{\phi}{Q_j(e)\psi}\nonumber\\
&& + O(1/t+\e^{w t}/t^2).
\end{eqnarray}
On the right side appears the $P_e$-block $P_e\Pi_{E_e}P_e$ of the projection $\Pi_{E_e}$. The contributions of the terms in \eqref{08} with ${\mathcal B}(z)$ and $R_{z}^{P_e}$ will add the remaining blocks to finally give the full expression $\scalprod{\phi}{\Pi_{E_e}\psi}$.

\subsubsection{The contribution to \eqref{08} from $\tfrac{d}{dz}{\mathcal B}(z)$} From \eqref{09'} and \eqref{033}, we have 
\begin{equation}
\label{036}
\tfrac{d}{dz}\scalprod{\phi}{{\mathcal B}(z)\psi} = \d \ \sum_{j=0}^{m_e-1} \frac{1-\d^2a'_j}{(z-e-\d^2a_j)^2}\, q_j+ \frac{1}{e-z+\d^2a_j} \, q'_j,
\end{equation}
where
\begin{equation}
\label{037}
q_j(z) =  - \scalprod{\phi}{[Q_jP_e I R_z^{P_e} + R_z^{P_e} I P_eQ_j -\d  R_z^{P_e} I P_e Q_j P_e I R_z^{P_e} ]\psi}.
\end{equation}
The summand in \eqref{036} is of the same form as $T_j$ in \eqref{09}, with $Q_j$ replaced by $q_j$. We may thus repeat the analysis leading to \eqref{018} and \eqref{029}, giving
\begin{eqnarray}
\label{038}
\frac{1}{\i t}\frac{1}{2\pi\i} \int_{\G_e}\e^{\i tz}\tfrac{d}{dz}\scalprod{\phi}{{\mathcal B}(z)\psi}dz &=& \e^{\i tE_e}\frac{\d\, q_0(E_e)}{1-\d^2 a'_0(E_e)} +\d\,\sum_{j=1}^{m_e-1} \e^{\i t(e+\d^2a_j(e))}q_j(e) \\
&&\ \ + C(\phi,\psi) \cdot \,O\left(|\d|/t+|\d|\e^{wt}/t^2\right).\nonumber
\end{eqnarray}
Recalling \eqref{031}, \eqref{032} and using that $R_{E_e-\i\epsilon}^{P_e}$ stays bounded as $\epsilon\rightarrow 0_+$, we get
\begin{eqnarray}
\scalprod{\phi}{\Pi_{E_e}\psi} &=& \lim_{\epsilon\rightarrow 0_+} (\i\epsilon) \scalprod{\phi}{\left\{\F(L-E_e+\i\epsilon;P_e)^{-1} +{\mathcal B}(E_e-\i\epsilon)\right\}\psi}\nonumber\\
&=& \scalprod{\phi}{P_e\Pi_{E_e}P_e\psi} + \frac{\d\,q_0(E_e)}{1-\d^2 a'_0(E_e)}.
\label{039}
\end{eqnarray}

\subsubsection{The contribution to \eqref{08} from $\tfrac{d}{dz}R_z^{P_e}$}  By assumption (A2), \eqref{nlap}, we have $\scalprod{\phi}{\tfrac{d}{dz}R_z^{P_e}\psi}\ple 1$ and so we get 
\begin{equation}
\label{040}
\frac{1}{\i t}\frac{1}{2\pi \i} \int_{\G_e} \e^{\i tz} \scalprod{\phi}{\tfrac{d}{dz}R_z^{P_e}\psi}d z \ple \frac{C(\phi,\psi)}{t}.
\end{equation}
Combining \eqref{08}, \eqref{035}, \eqref{038}, \eqref{039} and \eqref{040} gives \eqref{13.01}. This concludes the proof of Proposition \ref{propzero}.\hfill \qed

\subsection{Estimates for $z$ in a vicinity of an unstable eigenvalue $e$.}
\label{nestunstable}

The analysis is the same, actually somewhat easier, than the one presented in Section \ref{nestzero}. Indeed, for an unstable eigenvalue $e$, all the $\lambda_{e,j}$ have strictly positive imaginary part (see Assumption (A4), (1)). Therefore, we can proceed as in  Section \ref{nestzero}, and in \eqref{09}, all the terms $T_j$ (even for $j=0$) are now treated as above, after \eqref{018}. We immediately obtain the following result. 
\begin{prop}
\label{prope}
We have, for $\phi,\psi\in\mathcal D$,
\begin{equation}
\label{n13}
J_e(t) =  \sum_{j=0}^{m_e-1} \e^{\i t(e+\d^2  a_{e,j}(e))}\scalprod{\phi}{(Q_j(e)+\d\, \widetilde Q_j)\psi} +R_e(t) 
\end{equation}
where (recall that $w>0$ is the arbitrary parameter in \eqref{14})
$$
|R_e(t)|\le C \,\frac{1+\e^{wt}/t}{t},
$$
for a constant $C$ independent of $\d,t,w$ and 
\begin{equation}
\label{n13.1}
a_{e,j}(e) = \lambda_{e,j} +O(|\d|),\qquad Q_j(e)=P_{e,j}+O(|\d|).
\end{equation}
Also $\widetilde{Q}_j$ is defined in \eqref{qtilde}. 
\end{prop}

\subsection{Estimates for $z$ away from the eigenvalues $e$}

\label{inftysect}

On the unbounded set $\G_\infty$  we use the relation \eqref{i0}. This will help to ensure that the integrand is decaying sufficiently quickly at infinity. Setting $\widetilde\phi=(L-\i)\phi$ and $\widetilde\psi = (L+\i)\psi$, 
\begin{eqnarray}
\int_{\G_\infty} \e^{\i tz} \scalprod{\phi}{\tfrac{d}{dz}R_z\psi} dz  &=&  \int_{\G_\infty} \e^{\i tz}\left\{ \frac{-2\scalprod{\widetilde\phi}{R_z\widetilde\psi}}{(z+\i)^3}  +\frac{\scalprod{\widetilde\phi}{\tfrac{d}{dz}R_z\widetilde\psi}}{(z+\i)^2}  \right\} dz \label{i4}\\
&&+\int_{\G_\infty} \e^{\i tz}\left\{  \frac{\scalprod{\phi}{\psi}}{(z+\i)^2} + \frac{2\scalprod{\phi}{\widetilde\psi}}{(z+\i)^3}     \right\} dz. \label{i5}
\end{eqnarray}
The integral \eqref{i5} is $\ple \|\phi\|\,\|L\psi\|$. The terms on the right side of \eqref{i4} involving $R_z$ and $\tfrac{d}{dz}R_z$ are estimated using Assumption (A2)(2). Thus
\begin{equation}
\label{i6}
\frac{1}{\i t}\frac{1}{2\pi\i}\int_{\G_\infty} \e^{\i t z}\scalprod{\phi}{\tfrac{d}{dz} R_z\psi} dz \ple \frac{C(\widetilde\phi,\widetilde\psi)}{t}. 
\end{equation}
We have proven the following result. 
\begin{prop}
\label{propaway}
We have 
\begin{equation}
\label{13'''}
J_\infty(t) \ple \frac{C((L-\i)\phi,(L+\i)\psi)}{t}.
\end{equation}
\end{prop}

\subsection{Proof of Theorem \ref{nthm1}}

We combine the estimates in Propositions \ref{propzero}, \ref{prope} and \ref{propaway} to obtain the expansion \eqref{nmainreseqn} with the bound 
$$
|R(t)|\ple \frac{1+\e^{w t}/t}{t}
$$
for the remainder. Since $w>0$ is arbitrary, we have $|R(t)|\ple 1/t$. \hfill \qed

\appendix

\section{}
\label{appendix1}

\subsection{From the global to the local limiting absorption principles}

In this appendix, we derive the limiting absorption principle with projection $P_e$, $e\in{\rm spec}(L_\S)$, from that with projection $P_\R$ given in Theorem \ref{lapos}, \eqref{nlapos1}. Let 
$$
R_z^{P_e}= (P_e^\perp LP_e^\perp -z)^{-1}\upharpoonright_{{\rm Ran}P_e^\perp}
$$
and recall the definition of the energy gap $g$, \eqref{gapdef}. 
\begin{thm}
\label{lape}
Assume the conditions of Theorem \ref{lapos}. Then, for $\varphi,\psi\in\mathcal D$,
\begin{equation}
\label{lapeeqn}
\sup_{z\in{\mathbb C}_-\,:\, |z-e|\le g/2}| \tfrac{d^k}{dz^k} \scalprod{\phi}{R_z^{P_e}\psi}| = C(\phi,\psi)<\infty ,\quad k=0,\ldots, 3.
\end{equation}
Furthermore, \eqref{nlap1} holds.  Moreover, suppose that for $\phi,\psi\in{\mathcal D}$, $\sup_{z\in{\mathbb C}_-}|\frac{d^k}{d\d^k}\scalprod{\phi}{R_z^{P_\R}\psi}|\le C(\phi,\psi)$, $k=0,1$. Then, for all $e$,
\begin{equation}
\label{ddlocal}
\sup_{z\in{\mathbb C}_-}|\tfrac{d^k}{d\d^k}\scalprod{\phi}{R_z^{P_e}\psi}|\le C(\phi,\psi),\qquad k=0,1.
\end{equation}
\end{thm}
\noindent
(Note: the constants $C(\phi,\psi)$ may differ from the one in \eqref{nlapos1}.)

{\em Proof.\ } We first show \eqref{lapeeqn} by expressing  $R^{P_e}_z$ in terms of $R^{P_\R}_z$ and then using the bound \eqref{nlapos1}.  To do so, consider the operator
 \begin{equation}
K= P_e^\perp L P_e^\perp+ \i P_e,
\label{b1}
\end{equation}
where $\i$ is the imaginary unit. We have $(K-z)^{-1}=(\i-z)^{-1}P_e\oplus R^{P_e}_z\,P_e^\perp$ and therefore
\begin{equation}
\label{b4}
R_z^{P_e} =  (K-z)^{-1} P_e^\perp= P_e^\perp (K-z)^{-1} P_e^\perp.
\end{equation}
Next, denote
\begin{equation}
\label{b6}
\F_z \equiv \F(K-z;P_\R).
\end{equation}
Then by \eqref{mmblock} (with $Q=P_\R\equiv \bbbone_\S\otimes|\Omega_\R\rangle\langle\Omega_\S|$)
\begin{equation}
\label{b7}
(K-z)^{-1} =
\left(
\begin{array}{cc}
\one & 0\\
-\d R_z^{P_\R} P_\R^\perp I\bar P_e & \one
\end{array}
\right)
\left(
\begin{array}{cc}
\F_z^{-1} & 0\\
0 & R^{P_\R}_z 
\end{array}
\right)
\left(
\begin{array}{cc}
\one & -\d \bar P_e I P_\R^\perp R_z^{P_\R}\\
0 & \one
\end{array}
\right),
\end{equation}
where we have set 
\begin{equation}
\label{dueto}
\bar P_e = \bbbone[L_\S\neq e]\,P_\R = P_e^\perp P_\R. 
\end{equation}
Due to \eqref{dueto} and Theorem \ref{lapos}, $\forall \psi\in{\mathcal D}$ the maps $z\mapsto \bar P_eI P_\R^\perp R_z^{P_\R}\psi$ and $z\mapsto \langle\psi| R_z^{P_\R}P_\R^\perp I\bar P_e$ are three times differentiable with bounded derivatives for $z\in\overline{\mathbb C}_-$.

Next, since $P_e^\perp P_\R^\perp=  P_\R^\perp$ and $P_\R^\perp P_e= 0$, we have $(P_\R^\perp KP_\R^\perp-z)^{-1}\upharpoonright_{{\rm Ran} P_\R^\perp} =R^{P_\R}_z$ and so
\begin{equation}
\label{MM-eq1}
\F_z = (\i-z)\,P_e \oplus \bar P_e\big( L_\S-z +\d \, I-
\d^2\, I R^{P_\R}_z I  \big)\,\bar P_e.
\end{equation}
Since $P_\R IP_\R\ple 1$ and $P_\R IR_z^{P_\R}IP_\R\ple 1$, we have 
\begin{equation}
\label{b8}
\F_z^{-1} = (\i-z)^{-1} P_e\oplus (L_\S-z)^{-1} \bar P_e \big(\bbbone +O(\d)\big)
\end{equation}
provided $z\in{\mathbb C}_-$ and $|{\rm Re}z-e|\le \frac12 g$. The remainder term in \eqref{b8} is uniform in these $z$ and thus, for $\d$ small enough, $\F_z^{-1} \ple 1$. We obtain $\tfrac{d}{dz}\mathfrak{F}_z^{-1}=-\mathfrak{F}_z^{-1}\big[\tfrac{d}{dz}\mathfrak{F}_z\big]\mathfrak{F}_z^{-1}\ple 1$ and, taking further $z$-derivatives,
\begin{equation}
\label{b9}
\tfrac{d^k}{dz^k}\F_z \ple 1,\qquad k=0,\ldots,3.
\end{equation}
Combining \eqref{b4}, \eqref{b7}, \eqref{b9} and the regularity (first three $z$-derivatives bounded for $z\in\overline{\mathbb C}_-$) of the matrices to the left and right in \eqref{b7} discussed above yields the result \eqref{lapeeqn}.

To prove \eqref{nlap1}, the limiting absorption principle away from the eigenvalues of $L_0$, we apply the Feshbach map with projection $P_\R=\bbbone_\S\otimes|\Omega_\R\rangle\langle\Omega_\R|$,
\begin{equation}
\label{i1}
\F(z)\equiv \F(L-z,P_\R) = P_\R(L_\S+\d I -z-\d^2 IR_z^{P_\R} I)P_\R.
\end{equation}
 According to the decomposition \eqref{0}, 
\begin{equation}
\label{i8}
\scalprod{\phi}{\textstyle\frac{d}{dz}R_z\psi}  = \scalprod{\phi}{\textstyle\frac{d}{dz}\F(z)^{-1}\psi} + \scalprod{\phi}{\textstyle\frac{d}{dz}{\mathcal B}(z)\psi}+ \scalprod{\phi}{\textstyle\frac{d}{dz}R^{P_\R}_z\psi},
\end{equation}
where 
\begin{equation}
{\mathcal B}(z) = -\d \F(z)^{-1} P_\R I R_z^{P_\R} - \d R_z^{P_\R} I P_\R\F(z)^{-1} +\d^2 R_z^{P_\R} I P_\R \F(z)^{-1} P_\R I R_z^{P_\R} .
\label{i9}
\end{equation}
For any $z\in S_\infty$, we have $\|(L_\S-z)^{-1}\|\le 1/\alpha$ and therefore
\begin{equation}
\label{i2}
\F(z)^{-1} = (L_\S-z)^{-1} \left[ \bbbone +\d P_\R(I-\d IR_z^{P_\R}I)P_\R(L_\S-z)^{-1}\right]^{-1}\ple 1.
\end{equation}
We then obtain at once from \eqref{i2}, \eqref{i9} and \eqref{nlapos1} that  
\begin{equation}
\label{i2'}
\tfrac{d}{dz}\F(z)^{-1} \ple 1, \quad \scalprod{\phi}{\tfrac{d}{dz} {\mathcal B}(z)\psi}\le  |\d|\,   C(\phi,\psi), \quad \scalprod{\phi}{\tfrac{d}{dz} R_z^{P_\R}\psi}\ple   C(\phi,\psi).
\end{equation}
The bounds \eqref{i2'} together with \eqref{i8} imply \eqref{nlap1}. 

Finally, we prove \eqref{ddlocal}. We use again relations \eqref{b4} and \eqref{b7} to express $R_z^{P_e}$ in terms of $R_z^{P_\R}$. Taking the $\d$-derivative in \eqref{b7} results in taking $\d$-derivatives of $R_z^{P_\R}$ and $\F_z^{-1}$. The first one is controlled by the assumption in Theorem \ref{lape}, the second one is controlled by \eqref{MM-eq1} (as for the $z$-derivatives above). \hfill \qed

\subsection{Proof of Theorem \ref{MK-Lem1}}

For $\eta>0$, we introduce the regularized Liouville operator (see \cite{KoMeSo})
\begin{equation}
\label{c5}
\L(\eta) = \L_0+i\eta N +\Delta I(\eta),\qquad \mbox{with}\qquad  I(\eta)=(2\pi)^{-1/2}\int_{\mathbb R} \widehat f(s) \tau_{\eta s}( I)d s,
\end{equation}
and where $\tau_t(X)=e^{-it  A}Xe^{it  A}$ with $A=d\Gamma(i\partial_u)$. Here, $f$ is a Schwartz function satisfying $f^{(k)}(0)=1$, $k=0,1,\ldots$ $\L(\eta)$ is a closed operator on $\dom(\L_0)\cap\dom(N)$.

The strategy of the proof is to derive estimates \eqref{eq-w17}, \eqref{MK-eq19} for $\L$ replaced by the regularized $\L(\eta)$, namely,
\begin{eqnarray}
\sup_{z\in{\mathbb C}_-} |\tfrac{d^{\mu-1}}{dz^{\mu-1}}\scalprod{\phi}{ R^{ P_\R}_z(\eta)\psi}|&
\ple &  \|\phi\|_{\mu,2\mu}\,\|\psi\|_{\mu,2\mu}\label{eq-w17regularized} \\
\sup_{z\in{\mathbb C}_-}|\tfrac{d}{d\d}\scalprod{\phi}{R^{ P_\R}_z(\eta)\psi}|& \ple & \|\phi\|_{3,1}\|\psi\|_{3,1}.
\label{MK-eq19regularized}
\end{eqnarray}
Here, $\ple$ means \eqref{defple} with a constant not depending on $\eta>0$. 
We have $\tfrac{d^{\mu-1}}{dz^{\mu-1}}\scalprod{\phi}{ R^{ P_\R}_z(\eta)\psi}\rightarrow\tfrac{d^{\mu-1}}{dz^{\mu-1}}\scalprod{\phi}{ R^{ P_\R}_z\psi}$ and  $\tfrac{d}{d\d}\scalprod{\phi}{R^{ P_\R}_z(\eta)\psi}\rightarrow \tfrac{d}{d\d}\scalprod{\phi}{R^{ P_\R}_z\psi}$ as  $\eta\rightarrow 0_+$ (this follows from the fact that $R_z(\eta)\rightarrow R_z(0)$ strongly as $\eta\rightarrow 0_+$, see \cite{KoMeSo}, Lemma 4.3). Therefore,  \eqref{eq-w17}, \eqref{MK-eq19} follow from \eqref{eq-w17regularized}, \eqref{MK-eq19regularized}.

{\em Throughout this proof, we will not indicate the dependence of $\L$ and $I$ on $\eta$ and $\d$. We will also simply write $P$ instead of $P_\R$.}  In particular, $R_z^P=(\bar\L(\eta)-z)^{-1}$, where $\bar X=P^\perp_\R XP^\perp_\R\upharpoonright_{{\rm Ran} P^\perp_\R}$.

Let $X_\eta$ be an $\eta$-dependent bounded operator in $\mathcal{B}(\bar P \mathcal H)$ which is sufficiently regular in $\eta>0$. We define $\mathrm{ad}_A(X_\eta)=[A,X_\eta]$ and  $\partial X_\eta=\tfrac{d}{d\eta}X_\eta-\mathrm{ad}_A(X_\eta)$. Note that $\partial \L_0=0$ (recall that we write $\L_0\equiv \L_0(\eta)$). Since $\partial$ is a derivation we have, according to Leibniz' rule,
\begin{equation}
\partial  R^{P}_z=- \d R^{P}_z( \partial I) R^{P}_z, 
\label{MK-eq5}
\end{equation}
and for $\mu\ge 1$,
\begin{equation}
\partial  [R^{ P}_z]^\mu =-\d \sum_{j=1}^\mu [ R^{P}_z]^j (\partial  I)  [ R^{P}_z]^{\mu-j+1}\label{MK-eq5a}.
\end{equation}
Equations \eqref{MK-eq5} and \eqref{MK-eq5a} hold as equalities of bounded operators (note that $\partial I$ is relatively $N$-bounded and ${\rm ran} R_z^P\subset\dom(N)$, see \cite{KoMeSo}, Lemma 4.3).
 It follows that 
\begin{eqnarray}
\label{MK-eq6}
\tfrac{d}{d\eta}\scalprod{\phi}{ [R^{ P}_z]^\mu\psi}&=&\scalprod{\phi}{\mathrm{ad}_A  [R^{ P}_z]^\mu\psi}
+\scalprod{\phi}{ \partial  [R^{ P}_z]^\mu \psi}\\ \label{MK-eq6a}
&=&\scalprod{A\phi}{ [R^{ P}_z]^\mu\psi}-\scalprod{[R^{ P}_z]^{*\mu}\phi}{ A\psi}
+\scalprod{\phi}{ \partial  [R^{ P}_z]^\mu \psi}.
\end{eqnarray}
We are going to establish bounds on the right hand side. 
\begin{prop}[\cite{KoMeSo}] {\bf (1)} Suppose the form factor $f_\beta$ (see \eqref{2.9'''}) satisfies $\partial_u^jf_\beta\in L^2({\mathbb R}\times S^2,du\times d\Sigma)$, for $j=0,\ldots,\ell+1$.
Then, for $\ell\ge 1$,
\begin{equation}
\|\bar{N}^{-1/2}( \partial I) \bar{N}^{-\ell/2}\| \ple \eta^{\ell}.
\label{eq-w1}
\end{equation}
{\bf (2)} We have
\begin{align}
\|\bar{N}^{1/2}R^P_z \bar{N}^{1/2}\|&\ple \eta^{-1},\label{eq-w2}\\
\|\bar{N}^{1/2}R^P_z \psi\|&\ple \eta^{-1/2}\,\|\bar{N}^{-1/2}(1+A^2)^{1/2}\psi\|,\label{eq-w3}\\
\|\bar{N}^{\ell/2}R^P_z \psi\|&\ple \eta^{-1} \|\bar{N}^{(\ell-2)/2}R^P_z \psi\|+\eta^{-1} \|\bar{N}^{(\ell-2)/2}\psi\|, \ \ \ \ell\ge 2.
\label{eq-w4}
\end{align}
\end{prop}
In the arguments below in this proof, the biggest value of $\ell$ in \eqref{eq-w1} we will use is $\ell =2\mu$. Hence the regularity condition on $f_\beta$ in the Theorem \ref{MK-Lem1}. Combining \eqref{eq-w2} with \eqref{eq-w3} we get for $j=1,2,\ldots$
\begin{equation}\label{eq-w5}
\|\bar{N}^{1/2}[R^P_z]^j \psi\|\ple \eta^{-j+1/2}\|\bar{N}^{-1/2}(1+A^2)^{1/2}\psi\|.
\end{equation}
Moreover, from \eqref{eq-w4} we obtain for $j,\ell=1,2,\ldots$
\begin{equation}
\|\bar{N}^{\ell/2}[R^P_z]^j \psi\|\ple \eta^{-\left\lfloor\ell/2\right\rfloor}\sum_{k=1}^j\|\bar{N}^{1/2}[R^P_z]^k \psi\|
+\sum_{k=1}^{\left\lfloor\ell/2\right\rfloor}\eta^{-k} \|\bar{N}^{ (\ell-2k)/2}\psi\|.
\end{equation}
Since $\bar N\ge 1$ and because of \eqref{eq-w5} we obtain for $j,\ell=0,1,\ldots$
\begin{align}\label{eq-w7}
\|\bar{N}^{\ell/2}[R^P_z]^j \psi\|&\ple \eta^{-\left\lfloor\ell/2\right\rfloor-j+1/2}
\|\bar{N}^{-1/2}(1+A^2)^{1/2}\psi\|
+\eta^{-\left\lfloor\ell/2\right\rfloor} \|\bar{N}^{\ell/2}\psi\|.
\end{align}
From \eqref{MK-eq5a} and \eqref{eq-w1} follows that 
\begin{equation}\label{eq-w8}
|\scalprod{\phi}{\partial  [R^{ P}_z]^\mu\psi}|
\ple\eta^\ell \sum_{j=1}^\mu \|\bar{N}^{1/2} [ R^{P}_z]^{*j}\phi\| \|N^{\ell/2}  [ R^{P}_z]^{\mu-j+1}\psi\|.
\end{equation}
Since \eqref{eq-w5} holds with $[R^{ P}_z]$ replaced by $[R^{ P}_z]^*$,
we may apply  \eqref{eq-w5} and \eqref{eq-w7} to \eqref{eq-w8}
and obtain
\begin{eqnarray}\label{eq-w9}
|\scalprod{\phi}{\partial  [R^{ P}_z]^\mu\psi}|&\ple& \eta^\ell \sum_{j=1}^\mu 
\eta^{-j+1/2}\|\bar{N}^{-1/2}(1+A^2)^{1/2}\phi\|\\
\nonumber
&&\times 
\Big( \eta^{-\left\lfloor\ell/2\right\rfloor-\mu+j-1/2}
\|\bar{N}^{-1/2}(1+A^2)^{1/2}\psi\|
+\eta^{-\left\lfloor\ell/2\right\rfloor} \|\bar{N}^{\ell/2}\psi\|\Big)\\
\nonumber
&\ple&\eta^{\left\lceil \ell/2\right\rceil-\mu} \|\phi\|_{1,\ell}\,\|\psi\|_{1,\ell}.
\end{eqnarray}
By \eqref{MK-eq6a} we have
\begin{equation}
\label{eq-w14}
\tfrac{d}{d\eta} \scalprod{\phi}{ [R^{ P}_z]^\mu\psi}
=\scalprod{A\phi}{ [R^{ P}_z]^\mu\psi}-\scalprod{[R^{ P}_z]^{*\mu}\phi}{ A\psi}
+\scalprod{\phi}{ \partial  [R^{ P}_z]^\mu \psi}
\end{equation}
and hence 
\begin{eqnarray}
|\tfrac{d}{d\eta} \scalprod{\phi}{ [R^{ P}_z]^\mu\psi}|
& \ple&   \|\bar{N}^{-1/2}A\phi\|\,\|\bar{N}^{1/2}[R^{ P}_z]^\mu\psi\|+
\|\bar{N}^{1/2}[R^{ P}_z]^{*\mu}\phi\|\,\|\bar{N}^{-1/2}A\psi\|\\
\nonumber
&&+
\|\phi\|_{1,2\mu}\,\|\psi\|_{1,2\mu}.
\end{eqnarray}
By \eqref{eq-w5} we easily get 
\begin{equation}
\label{eq-w11}
|\tfrac{d}{d\eta} \scalprod{\phi}{ [R^{ P}_z]^\mu\psi}|
 \ple  \eta^{-\mu+1/2}\|\phi\|_{1,2\mu}\,\| \psi\|_{1,2\mu}.
\end{equation}
Since $\scalprod{\phi}{ [R^{ P}_z]^\mu\psi}=\scalprod{\phi}{ [R^{ P}_z(\eta=1)]^\mu\psi}
-\int_\eta^1\tfrac{d}{d\eta'} \scalprod{\phi}{ [R^{ P}_z(\eta')]^\mu\psi}d\eta'$, 
we conclude from \eqref{eq-w11} and from \eqref{eq-w5} for $\eta=1$ that
\begin{equation}
\label{eq-w12}
| \scalprod{\phi}{ [R^{ P}_z]^\mu\psi}|
 \ple (1+ \eta^{-\mu+3/2}) \|\phi\|_{1,2\mu}\,\|\psi\|_{1,2\mu}.
\end{equation}
Let us now consider again \eqref{eq-w14}, but instead of using \eqref{eq-w5} for an upper bound apply \eqref{eq-w12}
to $\scalprod{A\phi}{ [R^{ P}_z]^\mu\psi}$ and $\scalprod{\phi}{ [R^{ P}_z]^\mu A\psi}$. In this way we get for \eqref{eq-w14} the upper bound 
\begin{equation}
\label{eq-w15}
|\tfrac{d}{d\eta} \scalprod{\phi}{ [R^{ P}_z]^\mu\psi}|
 \ple (1+ \eta^{-\mu+3/2}) \|\phi\|_{2,2\mu}\,\|\psi\|_{2,2\mu}.
\end{equation}
Now we integrate (as before \eqref{eq-w12}) and obtain
\begin{equation}
\label{eq-w16}
| \scalprod{\phi}{ [R^{ P}_z]^\mu\psi}|
  \ple (1+ \eta^{-\mu+5/2}) \|\phi\|_{2,2\mu}\,\|\psi\|_{2,2\mu}.
\end{equation}
The right side of \eqref{eq-w16} has a power of $\eta$ reduced by one, as compared to the right side of \eqref{eq-w12}, 
but $\phi$ and $\psi$ contribute in a stronger norm. We continue this procedure to conclude the proof of \eqref{eq-w17regularized} as well as the bound
\begin{equation}
\sup_{z\in{\mathbb C}_-}| \tfrac{d}{d\eta}\tfrac{d^{\mu-1}}{dz^{\mu-1}}\scalprod{\phi}{ R^{ P_\R}_z\psi}|
  \ple  \|\phi\|_{\mu+1,2\mu}\,\|\psi\|_{\mu+1,2\mu}.
  \label{eq-w18}
\end{equation}
We now prove the bound \eqref{MK-eq19regularized}. Note that 
\begin{equation}
\label{ddelta}
\tfrac{d}{d\d}\scalprod{\phi}{R^{ P}_z\psi}=- \scalprod{\phi}{[R^{ P}_z I R^{ P}_z]\psi}.
\end{equation}
We have (see \eqref{MK-eq5})
\begin{equation}
\partial  [R^{ P}_z I R^{ P}_z]=-\d\,R^{ P}_z (\partial I) R^{ P}_z I R^{ P}_z+R^{ P}_z (\partial I)  R^{ P}_z
-\d\,R^{ P}_z I R^{ P}_z (\partial I) R^{ P}_z .
\end{equation}
Next, we use \eqref{eq-w1},\eqref{eq-w2} and $\|I\|\le 1$ to get
\begin{eqnarray}
|\scalprod{\phi}{\partial  [R^{ P}_z I R^{ P}_z]\psi}|
&\ple& \eta^{\ell-1}\|\bar{N}^{\ell/2}[R^{ P}_z]^*\phi\|\,\|R^{ P}_z\psi\|
+\eta^{\ell}\|\bar{N}^{\ell/2}[R^{ P}_z]^*\phi\|\,\|\bar{N}^{1/2}R^{ P}_z\psi\|\\
&&+\eta^{\ell-1}\|[R^{ P}_z]^*\phi\|\,\|\bar{N}^{\ell/2}R^{ P}_z\psi\|.
\end{eqnarray}
By \eqref{eq-w7} we have
\begin{equation}
\label{eq-w17'}
|\scalprod{\phi}{\partial  [R^{ P}_z I R^{ P}_z]\psi}|
\ple \eta^{\ell-1} \eta^{-\left\lfloor\ell/2\right\rfloor} \|\phi\|_{1,\ell}\|\psi\|_{1,\ell}.
\end{equation}
Recall that
\begin{equation}
\label{eqw-19}
\tfrac{d}{d\eta}\scalprod{\phi}{[R^{ P}_z I R^{ P}_z]\psi}=\scalprod{A\phi}{[R^{ P}_z I R^{ P}_z]\psi}-\scalprod{[R^{ P}_z I R^{ P}_z]^*\phi}{ A\psi}
+\scalprod{\phi}{\partial  [R^{ P}_z I R^{ P}_z]\psi}.
\end{equation}
Due to \eqref{eq-w3}, the first and the second expression on the right side
are bounded above by a constant times $ \eta^{-1} \|\phi\|_{2,0}\|\psi\|_{2,0}$. Thus we get from \eqref{eq-w17'}, \eqref{eqw-19}
\begin{equation}
|\tfrac{d}{d\eta}\scalprod{\phi}{[R^{ P}_z I R^{ P}_z]\psi}|\ple \eta^{-1} \|\phi\|_{2,1}\|\psi\|_{2,1}.
\end{equation}
We integrate from $\eta$ to 1 to obtain the estimate
\begin{equation}
\label{eq-w18'}
|\scalprod{\phi}{[R^{ P}_z I R^{ P}_z]\psi}|\ple |\ln(\eta)| \|\phi\|_{2,1}\|\psi\|_{2,1}.
\end{equation}
Next, we use \eqref{eq-w18'} in \eqref{eqw-19}  to get the better upper bound
\begin{equation}
|\tfrac{d}{d\eta}\scalprod{\phi}{[R^{ P}_z I R^{ P}_z]\psi}|\ple |\ln(\eta)| \|\phi\|_{3,1}\|\psi\|_{3,1}.
\end{equation}
Integration from $\eta$ to $1$ yields $|\scalprod{\phi}{[R^{ P}_z I R^{ P}_z]\psi}|\ple \|\phi\|_{3,1}\|\psi\|_{3,1}$, which, combined with \eqref{ddelta}, implies \eqref{MK-eq19regularized}. This completes the proof of Theorem \ref{MK-Lem1}. \hfill $\qed$

\section{}
\label{FeshbachAppendix}

Let $Q$ be an orthogonal projection and let $Q^\perp=\one-Q$. Let $H$ be a densely defined, closed operator satisfying ${\rm Ran} Q\subset \dom(H)$. Let $\bar{H}^Q$ be the operator on ${\rm Ran} Q^\perp$, given by $Q^\perp H Q^\perp$. For $z\not \in \sigma(\bar{H}^Q)$ we set $R^Q_z=(\bar{H}^Q-z)^{-1}$ and we assume that $\| R^Q_zHQ\|<\infty$ and $\|QHR^Q_z\|<\infty$. Then we define
\begin{equation}
\label{mm202}
\F_z\equiv \F(H-z;Q)=Q \big( H - z - H Q^\perp  R^Q_z Q^\perp H\big) Q.
\end{equation}

\begin{thm}[Feshbach-theorem, \cite{BFS2}]
\label{MM-App1}
Let $z\not \in \sigma(\bar{H}^Q)$.
\begin{enumerate}
\item We have
\begin{align*}
z\in \sigma(H)&\Leftrightarrow 0\in \sigma(\F_z)\\
z\in \sigma_p(H)&\Leftrightarrow \{0\}\not=\ker\F_z
\end{align*}
If $z\not \in \sigma(H)$ we have for $R_z=(H-z)^{-1}$ that
\begin{gather*}
QR_zQ=\F_z^{-1},\qquad Q^{\perp}R_zQ= -R^Q_z \,Q^{\perp} H Q\,\F_z^{-1},\qquad QR_zQ^{\perp}= -\F_z^{-1} \,Q\,H\, Q^{\perp}\,R^Q_z,\\
Q^\perp R_z Q^\perp = R^Q_z + R^Q_z \,Q^\perp\, H\, Q \F_z^{-1}Q\, H\, Q^\perp  R^Q_z.
\end{gather*}
\item If $Q$ has finite rank, we have the additional characterization 
\begin{equation*}
z\in \sigma(H)\Leftrightarrow z\in \sigma_p(H) \Leftrightarrow \det(\F_z)=0.
\end{equation*}
\item Let $\mathcal{H}_E$ be the eigenspace of $H$ for $E\in \sigma_p(H)$. The
restricted projection $Q|_{\mathcal{H}_E}$ is a bijection from $\mathcal{H}_E$
to $\ker \F_E$. Its inverse is
\begin{equation*}
\phi \mapsto \phi\oplus -R^Q_z Q^\perp H Q \phi.
\end{equation*}
\end{enumerate}
\end{thm}
The previous theorem is proven in \cite{BFS2}.  Point (1) of the theorem gives the block representation
\begin{eqnarray}
\label{mmblock}
R_z &=& 
\left(
\begin{array}{cc}
\F_z^{-1} & -\F_z^{-1} QH Q^\perp R_z^Q\\
-R_z^Q Q^\perp HQ \F_z^{-1} & R^Q_z +R^Q_z Q^\perp HQ\F_z^{-1} QHQ^\perp R_z^Q
\end{array}
\right)\\
&=&
\left(
\begin{array}{cc}
\one & 0\\
-R_z^Q Q^\perp HQ & \one
\end{array}
\right)
\left(
\begin{array}{cc}
\F_z^{-1} & 0\\
0 & R^Q_z 
\end{array}
\right)
\left(
\begin{array}{cc}
\one & -QH Q^\perp R_z^Q\\
0 & \one
\end{array}
\right).
\nonumber
\end{eqnarray}
Note that the three matrices on the last line are both invertible. Thus (recall \eqref{mm202})
\begin{equation}
\label{mmblock2}
H-z=\left(
\begin{array}{cc}
\one & QH Q^\perp R_z^Q\\
0 & \one
\end{array}
\right)
\left(
\begin{array}{cc}
\F(H-z;Q) & 0\\
0 & Q^\perp (H-z)Q^\perp
\end{array}
\right)
\left(
\begin{array}{cc}
\one & 0\\
R_z^Q Q^\perp HQ & \one
\end{array}
\right).
\nonumber
\end{equation}
{\em Remark.\ } If $Q'$ is another projection, satisfying $Q'Q=Q'=QQ'$, then
\begin{equation}
\label{mm204}
\F\big(\F(H-z;Q);Q'\big) = \F(H-z;Q').
\end{equation}
Indeed, from point (1) in the theorem, we know that 
$$
[\F(\F(H-z;Q);Q')]^{-1} = Q'[\F(H-z;Q)]^{-1}Q' = Q'Q R_z QQ' = Q' R_zQ' = [\F(H-z;Q')]^{-1},
$$
which implies \eqref{mm204}.

\begin{prop}[Weak Feshbach Theorem]
\label{Feshprop}
Assume that $H$ is self-adjoint,  $E\in \RR\setminus \sigma_p(\bar{H}^Q)$, and that
$\lim_{\epsilon\to 0_+}\F(H-E+\i\epsilon;Q)\equiv \F(H-E;Q)$ exists (as a weak limit). If $E$ is an eigenvalue of $H$ with eigenvector $\psi$, then $0$ is an eigenvalue of $\F(H-E;Q)$ with eigenvector $Q\psi$. 
\end{prop}
\begin{proof}
Applying the projections $Q$ and $Q^\perp$ to the eigenvalue equation, we get for any $\epsilon>0$
\begin{align}\label{m-12}
Q(H-E+\imath \epsilon)Q\psi + QHQ^\perp \psi&= \imath \epsilon\,Q\psi\\ \nonumber
Q^\perp \,H Q\psi + Q^\perp(H-E+\imath \epsilon)Q^\perp \psi&= \imath \epsilon \,Q^\perp \psi.
\end{align}
Applying the resolvent $R^Q_{E-\i\epsilon}$ to the second equality in \eqref{m-12} gives
\begin{equation}\label{m-11}
Q^\perp\psi= - R^Q_{E-\imath \epsilon}\, Q^\perp\,H\, Q\psi+ \imath \epsilon R^Q_{E-\imath \epsilon}Q^\perp \psi.
\end{equation}
Using \eqref{m-11} in the first equation of \eqref{m-12}, then letting $\epsilon\rightarrow 0_+$ and taking into account that $\imath \epsilon R^Q_{E-\imath \epsilon}Q^\perp \psi\to \one [\bar H^Q=E]\,Q^\perp \psi=0$, gives $\F(H-E;Q)\psi=0$. (Note that $Q\psi\neq 0$ for otherwise, one easily obtains from the second equality in \eqref{m-12} that $Q^\perp\psi=0$ as well.)
\end{proof}

\bigskip
\bigskip

\noindent
{\bf Acknowledgement.\ } This work has been supported by an NSERC Discovery Grant and an NSERC Discovery Grant Accelerator.

\end{document}